\documentclass{article} % For LaTeX2e
\usepackage{iclr2027_conference,times}
\usepackage{float}
\usepackage{amsmath,amsfonts,bm}

\def\eqref#1{equation~\ref{#1}}
\def\1{\bm{1}}

\DeclareMathAlphabet{\mathsfit}{\encodingdefault}{\sfdefault}{m}{sl}
\SetMathAlphabet{\mathsfit}{bold}{\encodingdefault}{\sfdefault}{bx}{n}

\usepackage{amsmath,amssymb,amsthm}
\newtheorem{proposition}{Proposition}[section]
\newtheorem{corollary}[proposition]{Corollary}
\newtheorem{theorem}[proposition]{Theorem}
\newtheorem{remark}[proposition]{Remark}
\usepackage{wrapfig}
\usepackage{hyperref}
\usepackage{url}
\usepackage{graphicx}
\usepackage{array}
\usepackage{tabularx}
\usepackage{enumitem}
\usepackage{booktabs}
\usepackage[table]{xcolor}
\usepackage{makecell}
\definecolor{categorygray}{HTML}{EEF1F4}
\definecolor{actsblue}{HTML}{DCEBFF}
\usepackage{arydshln}
\usepackage{wrapfig}
\usepackage{booktabs}
\usepackage{needspace}
\usepackage{algorithm}
\usepackage{algpseudocode}
\usepackage{float}
\usepackage{dblfloatfix}

\setlist[itemize]{
    leftmargin=1.3em,
    labelsep=0.45em,
    itemsep=0pt,
    topsep=2pt,
    parsep=0pt,
    partopsep=0pt
}

\title{Learning Who Talks to Whom: Role-Aware Communication Topology Scheduling for Efficient LLM Multi-Agent Collaboration}

\author{%
\textbf{Haibin Li}$^{*,1}$, \textbf{Wenya Yu}$^{*,2}$, \textbf{Xuefeng Jiang}$^{\dagger,3}$, \textbf{Jia Li}$^{3}$, \textbf{Xiao Han}$^{1}$, \textbf{Shuai Su}$^{2}$, \textbf{Wei Wei}$^{\dagger,1}$\\[-2pt]{\normalfont\small $^{1}$Beihang University}\\{\normalfont\small $^{2}$Nanyang Technological University}\\{\normalfont\small $^{3}$University of Chinese Academy of Sciences}
}

\iclrfinalcopy % suppresses ICLR review line numbers for the arXiv version.
\begin{document}

\maketitle\begingroup
\renewcommand{\thefootnote}{}
\footnotetext{\textsuperscript{*} Equal contribution; \textsuperscript{\ensuremath{\dagger}} Corresponding authors.}
\endgroup

\begin{abstract}
% xuefeng new version
Multi-agent systems driven by large language models (LLMs) rely on structured communication topologies to facilitate information exchange and collaborative reasoning. 
However, effectively scheduling these topologies to dynamically determine \textit{who talks to whom} and \textit{when to communicate} remains a fundamental challenge in balancing collective reasoning performance against exponential token overhead. 
Existing topology scheduling approaches either enforce static communication graphs that cannot adapt to dynamic problem-solving states, or overlook agent role complementarity while failing to explicitly incorporate the synthesized graph structure into team-level decision evaluation. 
To tackle these limitations, we propose \texttt{ACTS}, a role-aware adaptive communication topology scheduling framework for efficient LLM multi-agent collaboration. 
Orchestrated within a Centralized Training with Decentralized Execution (CTDE) paradigm, \texttt{ACTS} introduces a novel dual-attention mechanism. Specifically, a \textit{Role-conditioned Relational Attention} module dynamically selects variable-cardinality interaction partners based on agent roles and evolving execution states, while a \textit{Topological Attention Mixer} explicitly embeds the directed communication topology into team value decomposition for precise credit assignment.
Comprehensive evaluations across 9 challenging benchmarks covering code generation, multi-step reasoning, general knowledge, and competitive mathematics demonstrate that \texttt{ACTS} consistently outperforms state-of-the-art baselines, boosting average accuracy by up to 6.19 percentage points while reducing token consumption by over 15\%. The project website is available at \url{https://acts-mas.github.io/ACTS/}.

\end{abstract}

\section{Introduction}
\label{sec:introduction}

LLM-based multi-agent systems combine heterogeneous roles and complementary
reasoning perspectives to solve complex tasks \citep{llm_debate,autogen}.
Structured communication supports evidence sharing and cross-agent
verification but increases token consumption and inference latency
\citep{yan2025beyondselftalk}. Exchanges that overlook agent contributions
or evolving execution states may introduce redundant information
\citep{liu2024dylan,wang2025agentdropout}. This motivates role-aware adaptive
topology scheduling to determine when and between which agents communication
occurs, balancing task performance against communication cost
\citep{gdesigner,agent_qmix}.

\begin{figure*}[!t]
    \centering
    \includegraphics[
        width=0.88\textwidth,
        height=0.45\textheight,
        keepaspectratio
    ]{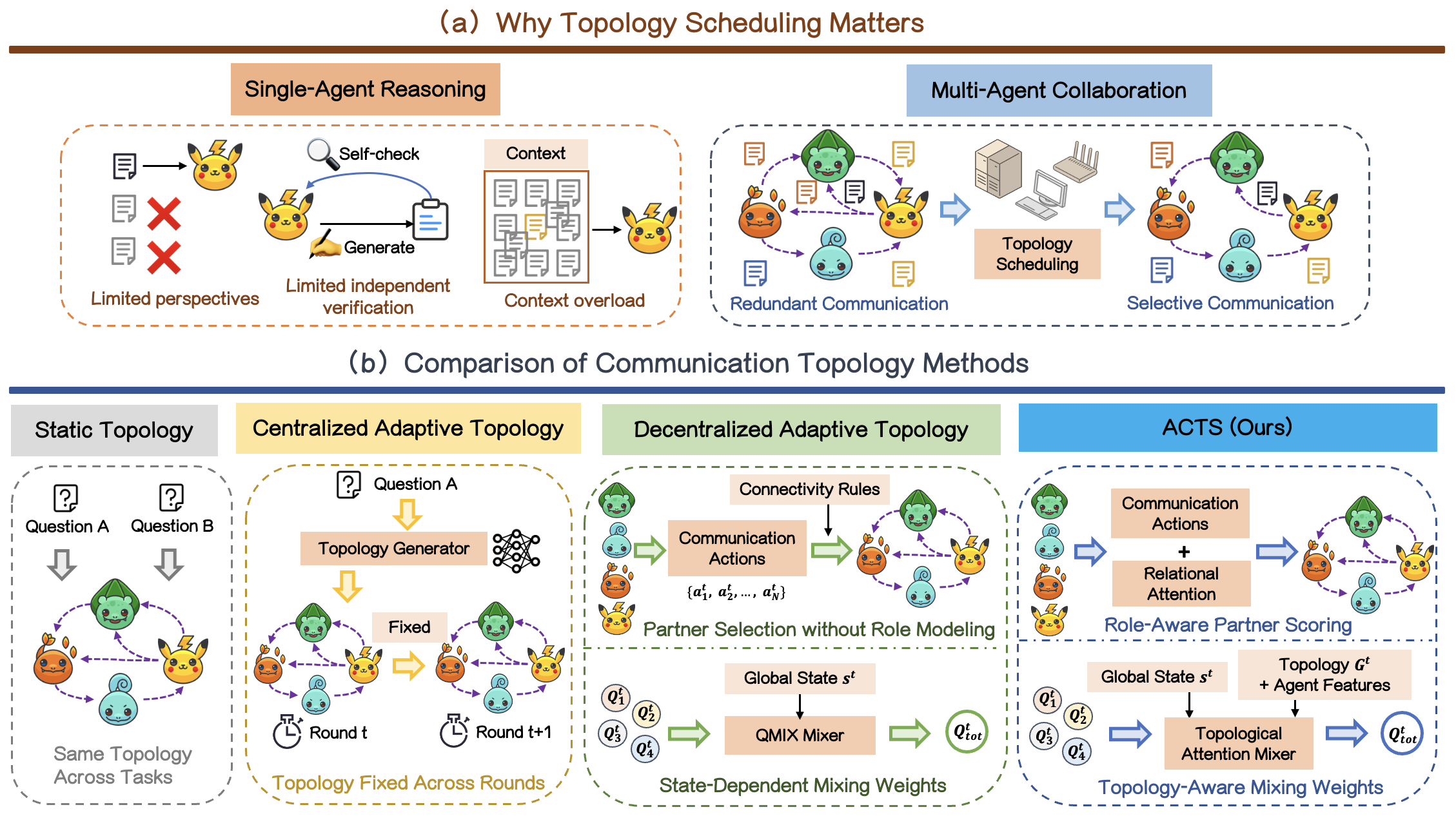}
    \caption{\textbf{Motivation and method comparison.}
(a) Topology scheduling for efficient multi-agent collaboration.
(b) Method comparison highlighting \texttt{ACTS}'s role-aware partner selection and topology-aware value mixing.}
    \label{fig:acts_intro}
\end{figure*}

Existing communication-topology scheduling methods can be classified into
static approaches, task-adaptive approaches, and multi-agent reinforcement learning (MARL)-based approaches. Static methods retain predefined interaction structures across
tasks and communication rounds
\citep{llm_debate,autogen}. Task-adaptive methods generate
task-specific graphs through topology induction or architecture search
\citep{gdesigner,topodim},
but commonly determine the complete topology before execution. Their directed
links consequently remain fixed during interaction and cannot adapt to
evolving local states. MARL-based topology-scheduling methods make communication decisions during
execution. Agent Q-Mix\citep{agent_qmix} is a representative MARL-based topology-scheduling method that
constructs directed topologies from selected communication modes through
predefined rules. Under centralized training with decentralized execution
(CTDE), a QMIX \citep{qmix} monotonic mixing network combines individual utilities using
global information during training, and agents select actions from local
observations during execution.
However, its partner selection is not explicitly conditioned on agent roles
or evolving execution states, and its mixer does not encode the realized
directed incoming and outgoing neighborhoods of individual agents. These
limitations constrain adaptive partner selection and topology-specific credit
assignment.

To address these limitations, this work proposes \texttt{ACTS}, a role-aware
framework for \textbf{A}daptive \textbf{C}ommunication
\textbf{T}opology \textbf{S}cheduling under centralized training with
decentralized execution (CTDE). Each agent uses a local action-value
function to decide \emph{when} and \emph{how} to communicate.
Role-aware relational attention selects variable-cardinality directed
partner sets from role embeddings and evolving local states, determining
who communicates with whom at each round. A topological attention mixer
incorporates the current directed topology into joint action-value
factorization through utility weights conditioned on the global state
and incoming and outgoing neighborhoods, enabling topology-aware temporal-difference (TD)
credit assignment. Under a fixed pre-decision context, the monotonic
mixer preserves Individual-Global-Max (IGM) consistency \citep{son2019qtran},
supporting decentralized greedy execution.

The main contributions are summarized as follows.

\begin{itemize}

\item \textbf{Role-aware communication topology scheduling.}
\texttt{ACTS} combines communication-mode selection with role-aware
relational attention to learn communication timing and variable-cardinality
directed partner sets from agent roles and evolving local states.
The scheduler jointly optimizes task performance and token efficiency.

\item \textbf{Topology-conditioned value factorization.}
A topological attention mixer learns nonnegative utility weights from
the global state and directed incoming and outgoing neighborhoods for
topology-aware TD credit assignment. The analysis relates graph-supported
interactions to local joint-value sensitivities and establishes
decentralized greedy consistency under a fixed pre-decision context.

\item \textbf{Accuracy and efficiency improvements.}
Evaluation covers nine benchmarks spanning four domains and two frozen
backbone LLMs. Relative to Agent Q-Mix, \texttt{ACTS} improves macro-average
accuracy by $3.18$ and $6.19$ percentage points and reduces average token
consumption by $15.42\%$ and $13.84\%$ on GPT-OSS-120B and
DeepSeek-V4-Flash, respectively. Ablations and communication-pattern
analyses further support the benefits of role embeddings, adaptive
routing, and topological mixing.

\end{itemize}

\section{Related Work}
\label{sec:related_work}

\subsection{Communication and Topology Control}
LLM-based multi-agent systems initially relied on fixed role-based workflows,
structured debates, and predefined interaction templates, including CAMEL,
AgentVerse, MetaGPT, ChatDev, AutoGen and LLM-Debate
\citep{li2023camel,chen2023agentverse,hong2023metagpt,qian2024chatdev,
autogen,llm_debate}. Subsequent methods adapt interaction
structures through graph induction, multi-agent architecture search, or
agent-participation control, such as G-Designer, GPTSwarm, MaAS, Assemble
Your Crew, AgentDropout, GTD, and TopoDIM
\citep{gdesigner,zhuge2024gptswarm,maas,li2025assemble,
wang2025agentdropout,gtd,topodim}.
These approaches generally do not jointly learn round-wise communication
modes, directed partner selection, and topology-conditioned team-value
evaluation. Agent Q-Mix introduces reinforcement-learning-based
round-wise communication scheduling with a QMIX mixer
\citep{agent_qmix}, but does not explicitly model agent roles or condition
its value factorization on individual directed communication neighborhoods.
\texttt{ACTS} addresses these limitations through role-aware
relational attention and topological attention for directed partner
selection and topology-aware value factorization.

% Existing LLM multi-agent systems predefine communication structures such as
% workflows, structured debates, and fixed interaction templates. CAMEL
% \citep{li2023camel}, AgentVerse \citep{chen2023agentverse}, MetaGPT
% \citep{hong2023metagpt}, ChatDev \citep{qian2024chatdev}, AutoGen
% \citep{autogen}, LLM-Debate \citep{du2023}, and DyLAN
% \citep{liu2024dylan} provide structured role-based collaboration but keep
% communication topologies fixed during execution. Other work instead adapts
% interaction structure to the task. G-Designer \citep{zhang2024gdesigner} and
% GPTSwarm \citep{zhuge2024gptswarm} induce interaction graphs, MaAS
% \citep{zhang2025maas} and Assemble Your Crew \citep{li2025assemble} optimize
% multi-agent architectures, AgentDropout \citep{wang2025agentdropout} modulates
% agent participation, and GTD \citep{jiang2025gtd} and TopoDIM
% \citep{sun2025topodim} generate complete topologies before execution. Agent
% Q-Mix \citep{agent_qmix} replaces handcrafted topology rules with reinforcement
% learning to select communication primitives and assemble round-wise
% topologies, but its QMIX mixer relies only on graph-level statistics and does
% not encode individual communication neighborhoods to determine mixing weights.
% \texttt{ACTS} addresses this limitation by using relational attention for role-aware
% partner selection and learning mixing weights from directed neighborhood
% representations through topological attention.

\subsection{Cooperative Multi-Agent Reinforcement Learning}

Cooperative MARL provides the optimization framework for decentralized
decision-making under a shared team objective. Centralized training with
decentralized execution (CTDE) uses global information during training while
restricting execution to local observations
\citep{oliehoek2016}. Value-factorization methods, including VDN, QMIX, QPLEX,
and Qatten, represent the team value through individual utilities with
different mixing structures
\citep{sunehag2018,qmix,wang2021qplex,yang2020qatten}. Differentiable
communication methods further study message passing and attention-based
aggregation \citep{sukhbaatar2016,foerster2016}. \texttt{ACTS} builds on this foundation
by learning discrete communication actions, role-aware interaction
topologies, and topology-dependent utility weights within a CTDE framework.

\section{ACTS: Dual-Attention Communication Topology Scheduling}
\label{sec:acts}
\texttt{ACTS} is a CTDE framework for adaptive communication-topology scheduling
in cooperative LLM agents \citep{oliehoek2016}. It combines value-based
mode selection with role-aware relational attention for directed partner
selection, and incorporates the resulting topology into centralized value
factorization through topological attention. Figure~\ref{fig:acts_pipeline}
illustrates the overall framework.

\begin{figure}[t]
    \centering
    \includegraphics[width=\linewidth]{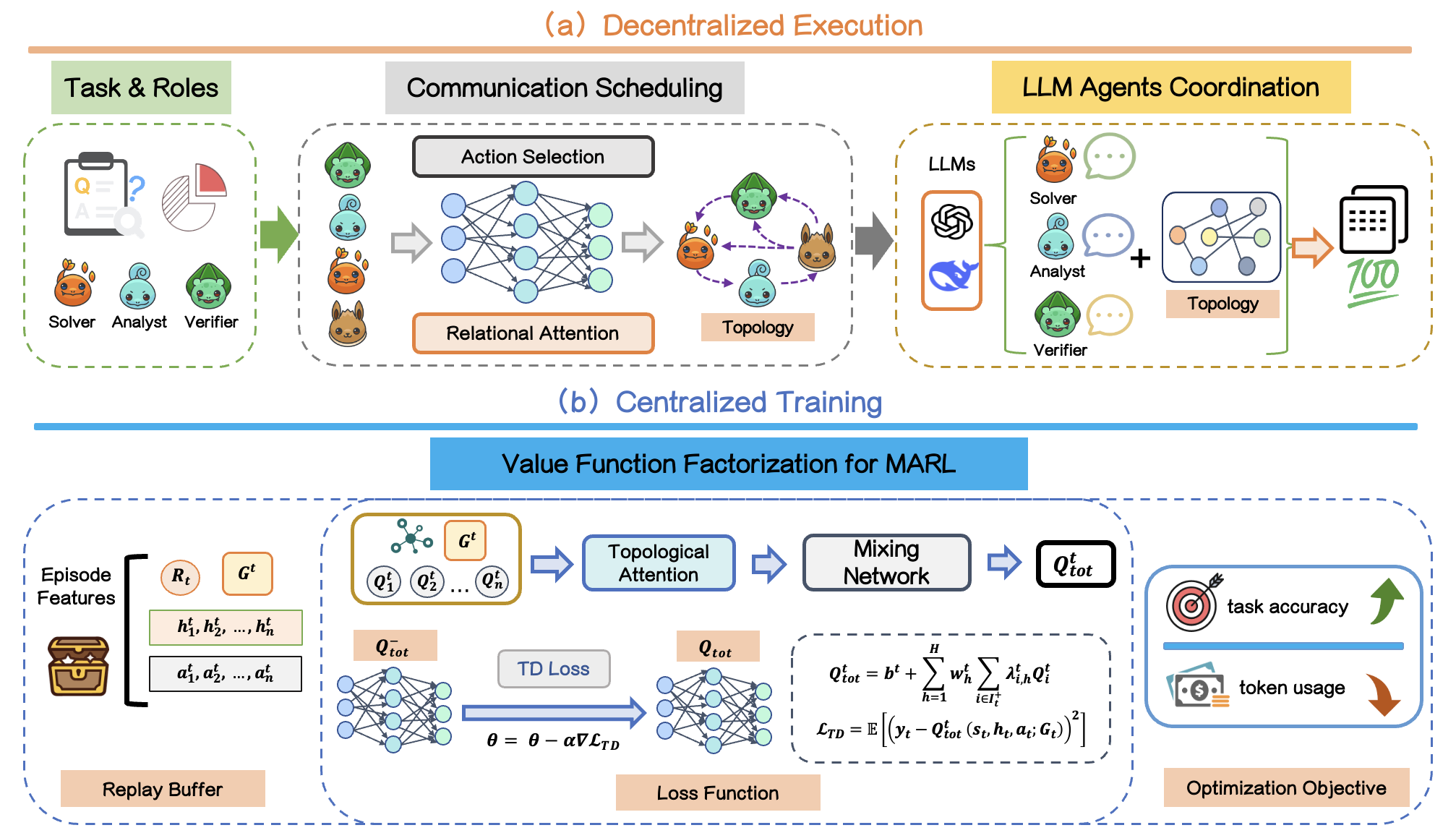}
    \caption{\textbf{Overview of \texttt{ACTS}.}
    \textbf{(a) Decentralized execution.}
    Local action-value functions select communication modes, while
    role-aware relational attention selects directed partners and
    induces the adaptive communication topology.
    \textbf{(b) Centralized training.}
    Episode transitions are stored in a replay buffer. Topological attention and
the mixing network combine the realized topology $G^t$ with local utilities
to estimate $Q_{\mathrm{tot}}^t$, which is optimized by TD learning under a
task-accuracy and token-cost objective.}
    \label{fig:acts_pipeline}
\end{figure}

\subsection{Problem Formulation}
\label{sec:problem_setup}

The proposed \texttt{ACTS} framework models adaptive communication topology
scheduling as a cooperative multi-agent decision problem under CTDE \citep{qmix}. Let
$\mathcal I=\{1,\ldots,N\}$ denote the set of cooperative LLM agents,
where agent $i$ is assigned a role $\rho_i$. At round $t$, agent $i$
receives a local observation $o_i^t$ and maintains a recurrent state
$h_i^t$ encoding its interaction history. The global state $s^t$ is
available only during centralized training.

Communication is represented by a directed graph
$G^t=(\mathcal I,\mathcal E^t)$, where $(i,j)\in\mathcal E^t$ denotes
a communication link from agent $i$ to agent $j$. We initialize
$G^1$ as the complete directed graph without self-loops and set
$\mathcal I_{\mathrm{comm}}^1=\mathcal I$. At round $t$, the agents
jointly select
$\boldsymbol a^t=(a_1^t,\ldots,a_N^t)$ with
$a_i^t\in\mathcal A_{\mathrm{comm}}$. The communication-eligible set
for the next round is
\begin{equation}
\mathcal I_{\mathrm{comm}}^{t+1}
=
\left\{
i\in\mathcal I
\mid
a_i^t\notin
\{a^{\mathrm{off}},a^{\mathrm{solo}}\}
\right\}.
\label{eq:acts_communication_eligible_set}
\end{equation}
The selected modes and directional attention scores then determine the sender and
receiver sets
$\mathcal S_i^{t+1}$ and $\mathcal R_i^{t+1}$, which define the directed
topology $G^{t+1}$.

For subsequent analysis, $\operatorname{IG}$ \citep{cover_thomas} denotes the
topology-induced information gain, which measures the
information-acquisition effectiveness of the topology scheduler. Let
$\pi$ denote the topology-scheduling policy and
$c_{\mathrm{tok}}^t(\tau)$ the token consumption at round $t$ along
trajectory $\tau$. The policy is optimized to balance task performance
and token efficiency:
\begin{equation}
\pi^\star
=
\operatorname*{arg\,max}_{\pi}
\mathbb E_{\tau\sim\pi}
\left[
w_{\mathrm{task}}R_{\mathrm{task}}(\tau)
-
w_{\mathrm{tok}}
\min\!\left(
\frac{\sum_{t=1}^{T}c_{\mathrm{tok}}^t(\tau)}{B},1
\right)
\right],
\label{eq:acts_scheduling_objective}
\end{equation}
where $R_{\mathrm{task}}(\tau)$ denotes the final task score, $B$ is a
token-cost normalization scale, and $w_{\mathrm{task}}$ and
$w_{\mathrm{tok}}$ control the trade-off between task performance and
token efficiency.

\subsection{Relational Attention for Communication Topology Learning}
\label{sec:topology_learning}
Existing topology schedulers rely on fixed communication patterns and do
not explicitly model role-dependent partner relevance
\citep{agent_qmix}. \texttt{ACTS} combines communication-mode selection with relational attention over
recurrent representations that encode agent roles and evolving local
contexts, thereby constructing adaptive directed neighborhoods with a
variable number of partners.

At round $t$, the GNN aggregates each agent's local observation and role
information according to the current topology $G^t$, while a GRU updates
the recurrent state using the previous hidden state:
\begin{equation}
h_i^t
=
\operatorname{GRU}_{\theta_h}\!\left(
\operatorname{GNN}_{\theta_g}\!\left(
o_i^t,e_{\rho}(\rho_i),G^t
\right),
h_i^{t-1}
\right),
\label{eq:acts_state_encoding}
\end{equation}
where $e_{\rho}(\rho_i)$ denotes the embedding of agent $i$'s role.
The resulting states are processed by a cross-attention module, which
computes the directed relevance between each ordered pair of agents:
\begin{equation}
\alpha_{ij}^t
=
\operatorname{softmax}_{j\neq i}
\left(
\frac{
\left(W_{Q,\mathrm{rel}}h_i^t\right)^\top
\left(W_{K,\mathrm{rel}}h_j^t\right)
}{
\sqrt{d_{\mathrm{rel}}}
}
\right).
\label{eq:acts_relational_attention}
\end{equation}

The ordered query--key construction produces directional attention
weights, with $\alpha_{ij}^t\neq\alpha_{ji}^t$ in general. The cross-attention layer refines each recurrent state by aggregating
value projections from other agents with a residual connection, producing
the features used for local mode selection:
\begin{equation}
\xi_i^t
=
h_i^t
+
\sum_{j\in\mathcal I\setminus\{i\}}
\alpha_{ij}^t W_{V,\mathrm{rel}}h_j^t,
\qquad
a_i^t
=
\operatorname*{arg\,max}_{a\in\mathcal A_{\mathrm{comm}}}
Q(\xi_i^t,a).
\label{eq:acts_cross_attention_selection}
\end{equation}

The communication mode specifies the direction and scope of interaction,
while relational attention selects the corresponding partners. For the
eligible candidate set $\mathcal J_i^{t+1}$, the receiver and sender sets
are
\begin{equation}
\begin{aligned}
\mathcal S_i^{t+1}
&=
\begin{cases}
\mathcal J_i^{t+1}, & a_i^t=a^{\mathrm{broadcast}},\\
\{j\in\mathcal J_i^{t+1}\mid
\alpha_{ij}^t>\bar{\alpha}_i^{\mathrm{out},t}\},
& a_i^t=a^{\mathrm{attn\text{-}broadcast}},\\
\varnothing, & \text{otherwise},
\end{cases}
\\[-1mm]
\mathcal R_i^{t+1}
&=
\begin{cases}
\mathcal J_i^{t+1}, & a_i^t=a^{\mathrm{receive}},\\
\{j\in\mathcal J_i^{t+1}\mid
\alpha_{ji}^t>\bar{\alpha}_i^{\mathrm{in},t}\},
& a_i^t=a^{\mathrm{attn\text{-}receive}},\\
\varnothing, & \text{otherwise}.
\end{cases}
\end{aligned}
\label{eq:acts_partner_selection}
\end{equation}

The resulting directed topology is
\begin{equation}
\mathcal E^{t+1}
=
\left\{
(i,j)\in\mathcal I\times\mathcal I
\;\middle|\;
i\neq j,\;
j\in\mathcal S_i^{t+1}
\ \lor\
i\in\mathcal R_j^{t+1}
\right\},
\qquad
G^{t+1}
=
\left(\mathcal I,\mathcal E^{t+1}\right).
\label{eq:acts_induced_topology}
\end{equation}

By incorporating role embeddings and topology-aware recurrent states,
\texttt{ACTS} conditions both action evaluation and partner selection on the
current agent context. This decoupled design enables directed,
variable-cardinality topologies that adapt across rounds. The complete
construction is given in Appendix~\ref{app:directed_topology_construction}.

\begin{figure}[t]
    \centering
    \includegraphics[width=\linewidth]{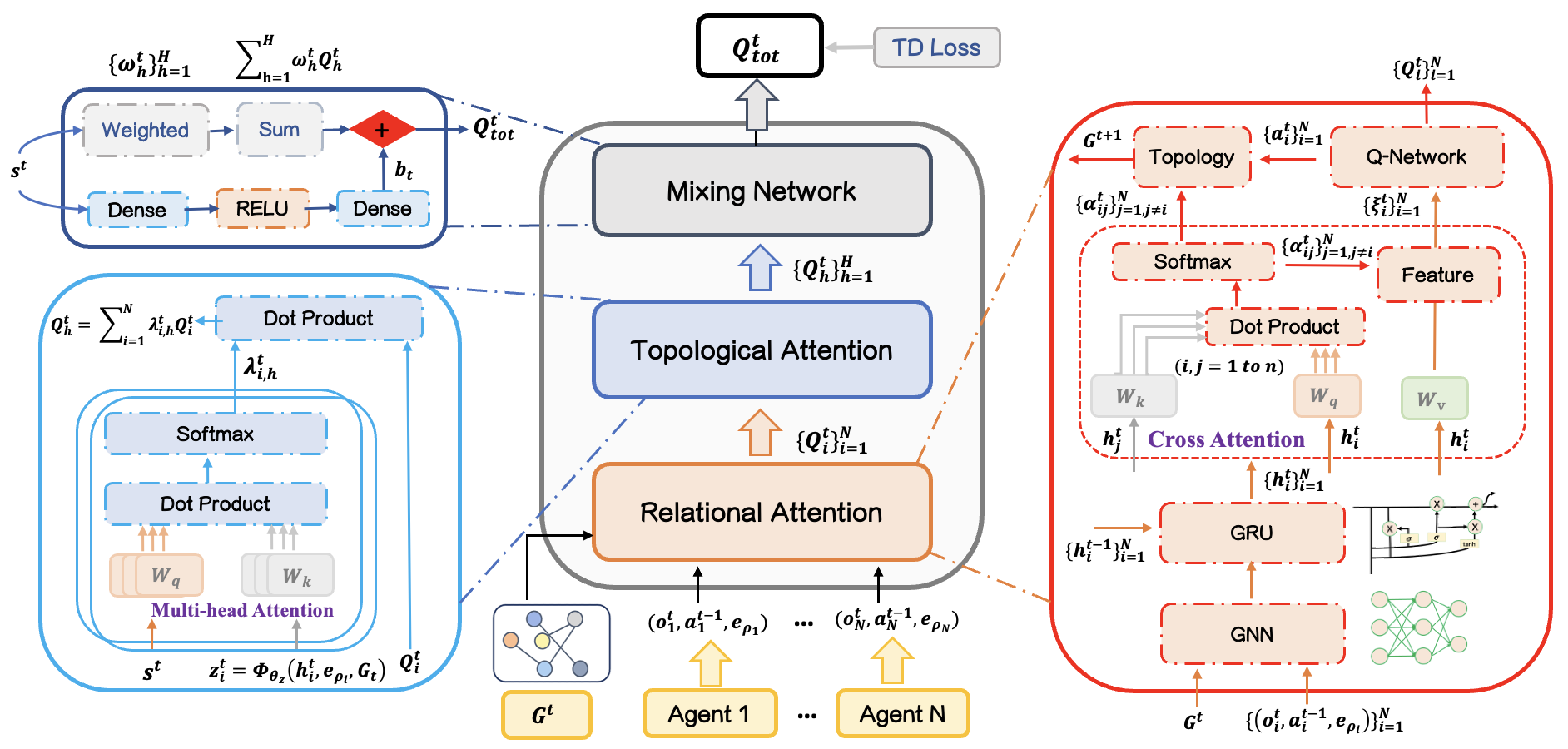}
    \caption{\textbf{Architecture of \texttt{ACTS}.}
Role-aware relational attention guides directed partner selection.
Topological attention weights local utilities according to directed
neighborhood representations, and the mixing network constructs
$Q_{\mathrm{tot}}^t$. These modules are jointly optimized by minimizing
the TD loss within a value function factorization MARL framework under CTDE.}
    \label{fig:acts_architecture}
\end{figure}

\subsection{Topological Attention for Value Function Factorization}
\label{sec:topological_factorization}
\texttt{ACTS} adopts CTDE value factorization with the Individual-Global-Max
property for decentralized greedy execution \citep{sunehag2018,qmix},
as reviewed in Appendix~\ref{app:marl_value_factorization}.
The QMIX mixer in Agent Q-Mix does not explicitly encode individual
directed communication neighborhoods. For a fixed pre-decision context,
Appendix~\ref{app:local_value_expansion} derives the local approximation
\begin{equation}
Q_{\mathrm{tot}}^t(\boldsymbol a)
\approx
b^t+
\sum_{i\in\mathcal I_{\mathrm{comm}}^t}
\beta_i^t Q_i(\xi_i^t,a_i),
\label{eq:acts_local_value_expansion}
\end{equation}
where $\beta_i^t$ denotes the local sensitivity to agent $i$'s utility
under the global state $s^t$. Under the graph interaction assumptions,
the derivation relates this sensitivity to interactions within the incoming
neighborhood $\mathcal N_i^{\mathrm{in},t}$ and outgoing neighborhood
$\mathcal N_i^{\mathrm{out},t}$. This motivates \texttt{ACTS} to learn
nonnegative mixing coefficients through topological attention conditioned
on the global state and separate representations of incoming and outgoing
neighborhoods.

For each candidate communication action $a_i$, \texttt{ACTS}
encodes the directed neighborhoods in $G^t$ by
\begin{equation}
\begin{aligned}
m_{i,r}^t
&=
\operatorname{Agg}_{r}
\left(
\left\{
\phi_r\!\left(
h_j^t\mathbin{\Vert}e_{\rho}(\rho_j)
\right)
\,\middle|\,
j\in\mathcal N_i^{r,t}
\right\}
\right),\\
z_i^t
&=
\Phi_{\theta_z}
\left(
h_i^t\mathbin{\Vert}e_{\rho}(\rho_i)
\mathbin{\Vert}
m_{i,\mathrm{in}}^t
\mathbin{\Vert}
m_{i,\mathrm{out}}^t
\right),
\qquad r\in\{\mathrm{in},\mathrm{out}\}.
\end{aligned}
\label{eq:acts_topology_encoder}
\end{equation}

The topology representation $z_i^t$, used only by the centralized mixer,
serves as the key in topological attention, while the global state $s^t$
provides the query. For each head $h$, we first compute the attention
logits and then normalize them over
$\mathcal I_{\mathrm{comm}}^t$:
\begin{equation}
\begin{aligned}
\eta_{i,h}^t
&=
\frac{
\left(W_{Q,\mathrm{top}}^{(h)}s^t\right)^\top
\left(W_{K,\mathrm{top}}^{(h)}z_i^t\right)
}{
\sqrt{d_h}
},
&
\lambda_{i,h}^t
&=
\left[
\operatorname{softmax}
\left(
\left(\eta_{j,h}^t\right)_{j\in\mathcal I_{\mathrm{comm}}^t}
\right)
\right]_i .
\end{aligned}
\label{eq:acts_topological_attention_weights}
\end{equation}
Thus,
$\lambda_{i,h}^t=
\exp(\eta_{i,h}^t)/
\sum_{j\in\mathcal I_{\mathrm{comm}}^t}\exp(\eta_{j,h}^t)$,
and the normalization is performed independently for each head.

With nonnegative head weights
$w_h^t=\left|[f_{\theta_w}(s^t)]_h\right|$ and offset
$b^t=V_{\theta_b}(s^t)$, the centralized mixer is
\begin{equation}
Q_{\mathrm{tot}}^t(\boldsymbol a)
=
b^t+
\sum_{h=1}^{H}
w_h^t
\left(
\sum_{i\in\mathcal I_{\mathrm{comm}}^t}
\lambda_{i,h}^t Q_i^t(a_i)
\right).
\label{eq:acts_topological_factorization}
\end{equation}

Here,
$\widehat{\beta}_i^t=\sum_{h=1}^{H}w_h^t\lambda_{i,h}^t$
is the effective mixing coefficient. Since
$w_h^t\geq0$ and $\lambda_{i,h}^t\geq0$,
$\partial Q_{\mathrm{tot}}^t/\partial Q_i^t(a_i)
=\widehat{\beta}_i^t\geq0$, so the mixer is monotonic in every local action
value. With the pre-decision context held fixed, nonnegative mixing ensures
that independent greedy actions maximize the learned joint value,
as established in Appendix~\ref{app:acts_igm}.
\begin{equation}
\prod_{i\in\mathcal I_{\mathrm{comm}}^t}
\operatorname*{arg\,max}_{a_i\in\mathcal A_i^t}
Q_i(\xi_i^t,a_i)
\subseteq
\operatorname*{arg\,max}_{\boldsymbol a\in
\prod_{i\in\mathcal I_{\mathrm{comm}}^t}\mathcal A_i^t}
Q_{\mathrm{tot}}^t(\boldsymbol a).
\label{eq:acts_igm}
\end{equation}
The inclusion becomes equality when the attention logits are finite
and $\sum_{h=1}^{H}w_h^t>0$.

\subsection{Cooperative MARL Training and Execution}
\label{sec:training_execution}

The training and execution procedures are detailed in
Appendix~\ref{app:acts_pseudocode}. During decentralized execution,
agent $i$ selects
$a_i^t=\operatorname*{arg\,max}_{a\in\mathcal A_{\mathrm{comm}}}
Q_i(\xi_i^t,a)$, and the joint action
$\boldsymbol a^t=(a_1^t,\ldots,a_N^t)$ induces the next-round topology
$G^{t+1}$. During centralized training, the mixer evaluates each
replayed transition conditioned on its pre-decision state and topology,
with the terminal team reward specified by ~\eqref{eq:acts_scheduling_objective}.

For each sampled transition, the target action of each communication-eligible
agent is selected greedily:
\begin{equation}
a_i^{*,t+1}
=
\operatorname*{arg\,max}_{a\in\mathcal A_{\mathrm{comm}}}
Q_i^{-}(\xi_i^{t+1},a),
\qquad
i\in\mathcal I_{\mathrm{comm}}^{t+1}.
\label{eq:acts_target_action}
\end{equation}

The TD target and loss are then defined as
\begin{equation}
\begin{aligned}
y^t
&=
r^t+\gamma Q_{\mathrm{tot}}^{-}
\left(\boldsymbol a^{*,t+1}\right),\\
\mathcal L_{\mathrm{TD}}
&=
\mathbb E_{\mathcal D}
\left[
\left(
y^t-Q_{\mathrm{tot}}^t(\boldsymbol a^t)
\right)^2
\right].
\end{aligned}
\label{eq:acts_td_update}
\end{equation}

The GNN, GRU, relational-attention module, local action-value networks,
topology encoder, and topological mixer are optimized jointly, while the
backbone LLM remains frozen. For terminal transitions, the target is $y^t=r^t$. Mixer coefficients are computed from the fixed
pre-decision context, and the discrete topology transition is excluded from
gradient computation.

\section{Experiments}
\label{sec:experiments}

\subsection{Experimental Setup}
\label{sec:experimental_setup}
We evaluate \texttt{ACTS} on nine benchmarks spanning four domains:
code generation (LiveCodeBench v6 and HumanEval), multiple-choice
reasoning (MMLU-Pro and the first 250 questions of Humanity's Last Exam),
multi-hop question answering (HotpotQA), and competition mathematics
(AIME 2025/2026, HMMT 2025, and Beyond-AIME)
\citep{jain2024livecodebench,chen2021humaneval,wang2024mmlupro,
hle2025,yang2018hotpotqa,maa2025aime,maa2026aime,matharena2025hmmt,
bytedance2025beyondaime}. We use GPT-OSS-120B \citep{openai2025gptoss} and DeepSeek-V4-Flash \citep{deepseek2026v4flash} as
frozen backbone LLMs; each task instantiates three role-conditioned agents
and a final answer-aggregation node. Baselines include direct prompting,
fixed or framework-based coordination (LLM-Debate, AutoGen, LangGraph, and
Agent Framework), task-adaptive topology scheduling (TopoDIM), and
MARL-based topology scheduling (Agent Q-Mix)
\citep{llm_debate,autogen,langchain2024langgraph,topodim,
agent_qmix}. All methods use the same prompts, team size,
communication-round budget, and evaluation protocol. Results are averaged
over three independent random seeds. Detailed benchmark, baseline, and
training configurations are provided in Appendix~\ref{sec:appendix_training}.

% We evaluate \texttt{ACTS} on nine benchmarks spanning four domains: code generation
% (LiveCodeBench v6 \citep{jain2024livecodebench}, HumanEval \citep{chen2021humaneval}),
% multiple-choice reasoning (MMLU-Pro \citep{wang2024mmlupro}, the first 250
% questions of Humanity's Last Exam \citep{hle2025}), multi-hop question answering
% (HotpotQA \citep{yang2018hotpotqa}), and competition mathematics (AIME 2025
% and 2026 \citep{maa2025aime}, HMMT 2025 \citep{matharena2025hmmt},
% Beyond-AIME \citep{bytedance2025beyondaime}). All experiments employ
% GPT-OSS-120B and DeepSeek-V4-Flash as frozen backbone LLMs, with each task
% instantiated as three role-conditioned agents followed by a final answer
% aggregation node. We compare \texttt{ACTS} against direct prompting, the static
% multi-agent coordination method LLM-Debate \citep{du2023}, the adaptive
% topology method TopoDIM \citep{sun2025topodim}, the general-purpose frameworks
% AutoGen \citep{autogen} and LangGraph \citep{langchain2024langgraph},
% an Agent Framework baseline, and the MARL-based approach Agent Q-Mix
% \citep{agent_qmix}. These are the methods for which experimental
% results are currently available. For a fair comparison, all methods share the
% same team size, task prompts, communication-round budget, and evaluation
% procedure. Training hyperparameters and further evaluation details are
% provided in Appendix~\ref{sec:appendix_training}. We report average
% performance over three independent runs with different random seeds.

\begin{table*}[t]
\centering
\caption{Accuracy (\%) across nine benchmarks. Bold marks the best result
in each column; \texttt{ACTS} is highlighted in light blue. Avg.\ is the unweighted
macro-average over the nine benchmarks.}
\label{tab:main_accuracy}

% \footnotesize
% \setlength{\tabcolsep}{2.5pt}
% \renewcommand{\arraystretch}{1.08}
\footnotesize
\setlength{\tabcolsep}{2.2pt}
\renewcommand{\arraystretch}{0.93}

\setlength{\abovetopsep}{2pt}
\setlength{\belowbottomsep}{0pt}
\setlength{\aboverulesep}{1.2pt}
\setlength{\belowrulesep}{1.2pt}

\begin{tabular}{l*{10}{c}}

\toprule

\rowcolor{gray!20}
Method
& \makecell{LiveCode\\Bench}
& \makecell{Human\\Eval}
& HLE
& \makecell{Hotpot\\QA}
& \makecell{MMLU-\\Pro}
& \makecell{AIME\\2025}
& \makecell{AIME\\2026}
& \makecell{Beyond-\\AIME}
& \makecell{HMMT\\2025}
& Avg. \\

\midrule
\rowcolor{white}
\multicolumn{11}{c}{\textit{LLM backbone: GPT-OSS-120B}} \\
\midrule

Direct Prompting
& 88.57 & 89.63 & 42.40 & 20.00 & 70.00
& 60.00 & 86.67 & 16.00 & 46.67 & 57.77 \\

\midrule
\multicolumn{11}{l}{\textit{Static topology}} \\
\midrule

LLM-Debate
& 92.00 & 88.41 & 42.80 & 21.35 & 71.43
& 63.33 & 86.67 & 18.00 & 46.67 & 58.96 \\

AutoGen
& 95.43 & 90.24 & 41.20 & 32.03 & 77.14
& 70.00 & 80.00 & 18.00 & 46.67 & 61.19 \\

LangGraph
& 93.14 & 89.02 & 42.40 & 40.00 & 84.29
& 76.67 & 86.67 & 23.00 & 50.00 & 65.02 \\

Agent Framework
& 93.14 & 90.24 & 42.00 & 32.03 & 84.29
& 73.33 & 76.67 & 36.00 & 40.00 & 63.08 \\

\midrule
\multicolumn{11}{l}{\textit{Centralized adaptive topology}} \\
\midrule

G-Designer
& 94.29 & 93.29 & 42.80 & 34.73 & 70.00
& 73.33 & 83.33 & 28.00 & 43.33 & 62.57 \\

TopoDIM
& 93.71 & 92.68 & 42.00 & 30.68 & 65.71
& 66.67 & 83.33 & 25.00 & 33.33 & 59.23 \\

GTD
& 94.86 & 93.90 & 42.40 & 35.95 & 71.43
& 70.00 & 80.00 & 30.00 & 46.67 & 62.80 \\

MaAS
& 94.86 & 93.29 & 41.60 & 33.38 & 68.57
& 66.67 & 86.67 & 27.00 & 40.00 & 61.34 \\

\midrule
\multicolumn{11}{l}{\textit{Decentralized learned topology}} \\
\midrule

Agent Q-Mix
& 95.43 & \textbf{100.00} & 38.40 & 40.00 & \textbf{100.00}
& 80.00 & 90.00 & 36.00 & 66.67 & 71.83 \\

\rowcolor{actsblue}
\textbf{ACTS}
& \textbf{96.57} & \textbf{100.00} & \textbf{43.20}
& \textbf{46.67} & \textbf{100.00}
& \textbf{83.33} & \textbf{93.33} & \textbf{42.00}
& \textbf{70.00} & \textbf{75.01} \\

\midrule

\multicolumn{11}{c}{\textit{LLM backbone: DeepSeek-V4-Flash}} \\
\midrule

Direct Prompting
& 93.14 & 98.17 & 33.60 & 40.00 & 72.86
& 50.00 & 66.67 & 10.00 & 46.67 & 56.79 \\

\midrule
\multicolumn{11}{l}{\textit{Static topology}} \\
\midrule

LLM-Debate
& 93.14 & 96.95 & 34.00 & 40.00 & 61.43
& 53.33 & 66.67 & 12.00 & 36.67 & 54.91 \\

AutoGen
& \textbf{96.00} & 98.17 & 32.80 & 42.70 & 71.43
& 60.00 & 70.00 & 15.00 & 46.67 & 59.20 \\

LangGraph
& 92.57 & 97.56 & 33.60 & 46.62 & 78.57
& 60.00 & 70.00 & 26.00 & 46.67 & 61.29 \\

Agent Framework
& 89.71 & 98.17 & 33.20 & 44.05 & 75.71
& 66.67 & 76.67 & 20.00 & 43.33 & 60.83 \\

\midrule
\multicolumn{11}{l}{\textit{Centralized adaptive topology}} \\
\midrule

G-Designer
& 93.71 & 97.56 & 33.60 & 46.62 & 88.57
& 63.33 & 73.33 & 22.00 & 43.33 & 62.45 \\

TopoDIM
& 92.57 & 97.56 & 32.80 & 44.05 & 88.57
& 60.00 & 70.00 & 18.00 & 46.67 & 61.14 \\

GTD
& 93.71 & 97.56 & 34.00 & 47.97 & 90.00
& 63.33 & 73.33 & 24.00 & 43.33 & 63.03 \\

MaAS
& 94.29 & 98.17 & 32.80 & 46.62 & 88.57
& 63.33 & 73.33 & 20.00 & 43.33 & 62.27 \\

\midrule
\multicolumn{11}{l}{\textit{Decentralized learned topology}} \\
\midrule

Agent Q-Mix
& 93.14 & \textbf{100.00} & 32.80 & 44.00 & 92.86
& 60.00 & 73.33 & 33.00 & 40.00 & 63.24 \\

\rowcolor{actsblue}
\textbf{ACTS}
& \textbf{96.00} & \textbf{100.00} & \textbf{34.40}
& \textbf{50.67} & \textbf{97.14}
& \textbf{70.00} & \textbf{86.67} & \textbf{40.00}
& \textbf{50.00} & \textbf{69.43} \\

\bottomrule
\end{tabular}

\end{table*}

\subsection{Main Results}
\label{sec:experimental_results}

\textbf{Accuracy gains from selective communication.} Table~\ref{tab:main_accuracy} reports
accuracy across all nine benchmarks. \texttt{ACTS} achieves the highest average accuracy on both
backbone models, $75.01\%$ on GPT-OSS-120B and $69.43\%$ on DeepSeek-V4-Flash, exceeding
Agent Q-Mix by $3.18$ and $6.19$ points in each case, and attains the best or tied-best
result on every benchmark. This advantage is concentrated on tasks where agents hold
complementary information, because relational attention selects communication partners by
semantic relevance rather than by fixed topology or content-agnostic index rules.

\textbf{Pareto-optimal efficiency across IG regimes.} Figure~\ref{fig:pareto} plots accuracy
against token consumption on HotpotQA ($\mathrm{IG}>0$) and HLE ($\mathrm{IG}\approx 0$);
full counts are provided in Table~\ref{tab:token_usage} (Appendix~\ref{sec:appendix_token}).
Static topologies broadcast every round regardless of whether communication is useful;
centralized adaptive topologies adjust the structure only at generation time, not during
execution. Agent Q-Mix learns agent deactivation and primitive-based routing, but discrete
communication primitives provide limited selectivity. \texttt{ACTS} achieves Pareto optimality
across all four settings: it matches or exceeds Agent Q-Mix at lower token cost and
dominates static and centralized families on both accuracy and cost. Only single-agent
prompting uses fewer tokens, at an accuracy cost of approximately $27$ percentage points
on HotpotQA with GPT-OSS-120B. On HLE, where information gain is near zero, \texttt{ACTS} and
single-agent prompting yield comparable performance. By learning whether communication is
warranted and routing through attention affinities, \texttt{ACTS} spends tokens only where expected
information gain justifies the expenditure.

\begin{figure}[htbp]
\centering
\includegraphics[width=1.0\linewidth]{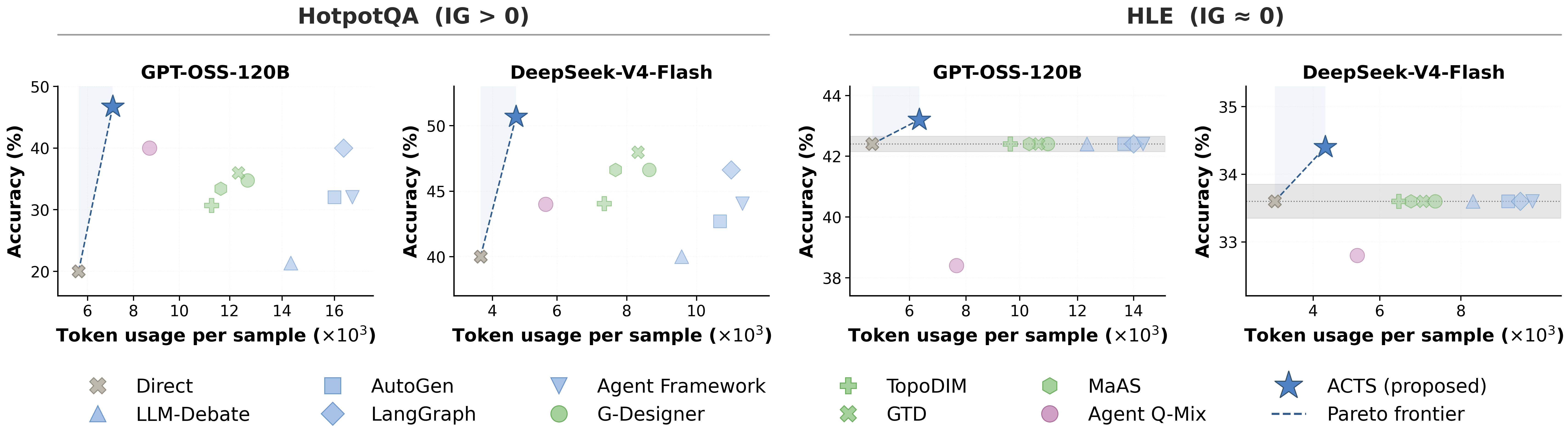}
\caption{Accuracy vs.\ token usage for all methods on two representative
tasks (HotpotQA with $\mathrm{IG}>0$ and HLE with $\mathrm{IG}\approx 0$)
and two backbone models. The dashed line traces the Pareto frontier of
methods not dominated on both accuracy and cost. \texttt{ACTS} lies on the
frontier in all four panels: no other multi-agent method reaches its
accuracy at lower or equal cost, and only single-agent prompting is
cheaper, at a substantial accuracy loss.}
\label{fig:pareto}
\end{figure}

\subsection{Experimental Analysis}
\label{sec:experimental_analysis}

\subsubsection{Information Gain Analysis}
\label{sec:info_gain}

\setlength{\intextsep}{4pt}

\begin{wrapfigure}{r}{0.48\columnwidth}
\vspace{-8pt}
\centering
\includegraphics[width=\linewidth]{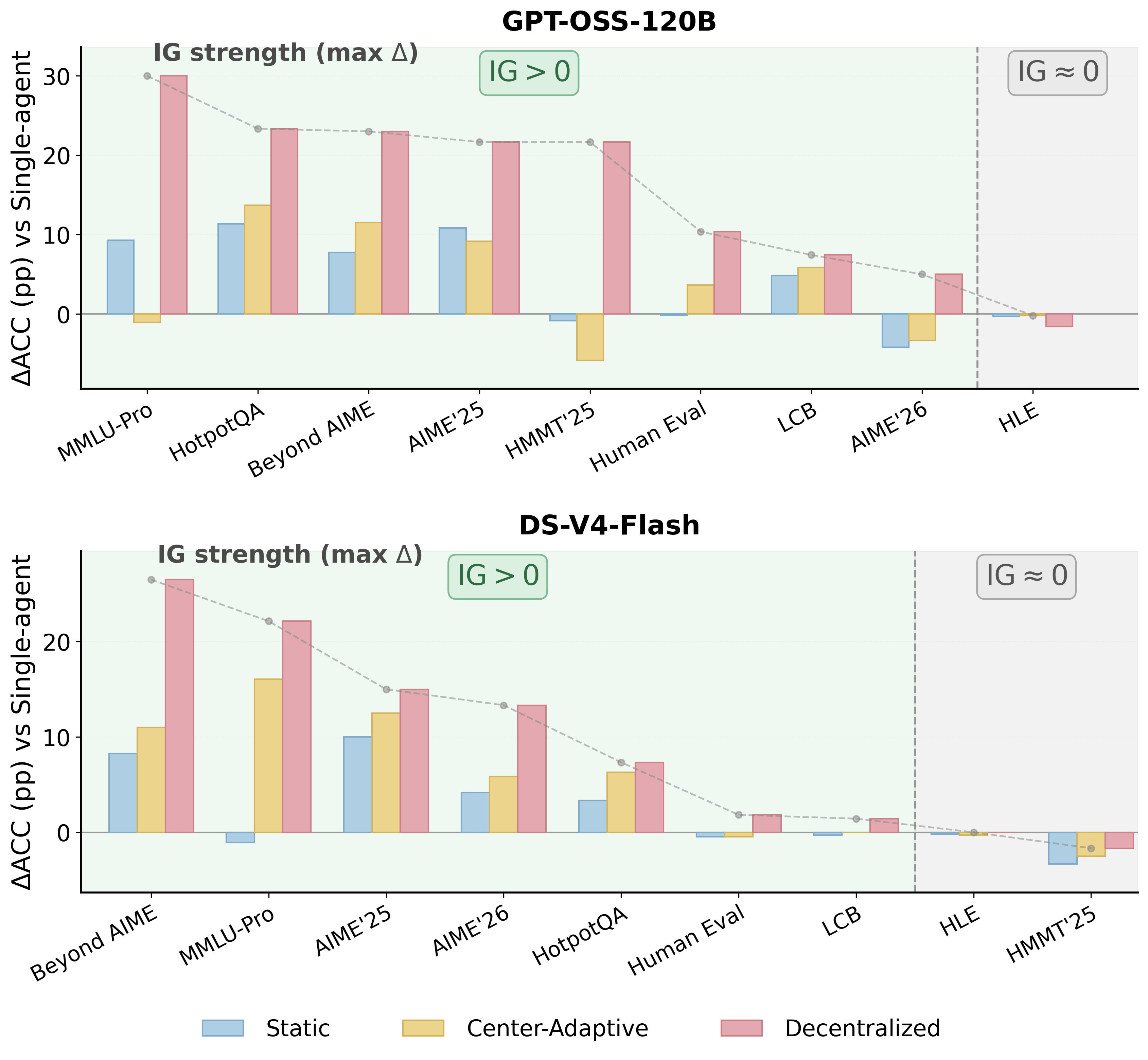}
\caption{Mean $\Delta$ACC of each topology family over single-agent
  prompting, by backbone model. The dashed line partitions benchmarks into
  $\mathrm{IG}>0$ on the left and $\mathrm{IG}\approx 0$
  within $\pm 2/n$ on the right.}
\label{fig:delta_acc}
\vspace{-12pt}
\end{wrapfigure}

Information gain (IG) \citep{cover_thomas} governs whether communication confers an advantage: agents benefit only when they contribute complementary evidence.
  Figure~\ref{fig:delta_acc} compares mean accuracy gains across topology families relative to single-agent prompting. Uplift falls below the
  sampling-noise threshold $x = 2/n$ on HLE and with limited-capacity DeepSeek-V4-Flash on HMMT, yielding null IG. It surpasses this threshold on
  HotpotQAand contest math, where cross-entity evidence and divergent solution paths engender separable signals, yielding positive IG. \texttt{ACTS}
  navigates both regimes: attention-driven partner selection amplifies communication when IG is positive; single-agent inference when null.

\subsubsection{Communication Pattern Analysis}
\label{sec:communication_analysis}

\setlength{\intextsep}{4pt}

\begin{wrapfigure}[15]{r}{0.48\columnwidth}
    \centering
    \includegraphics[width=\linewidth]{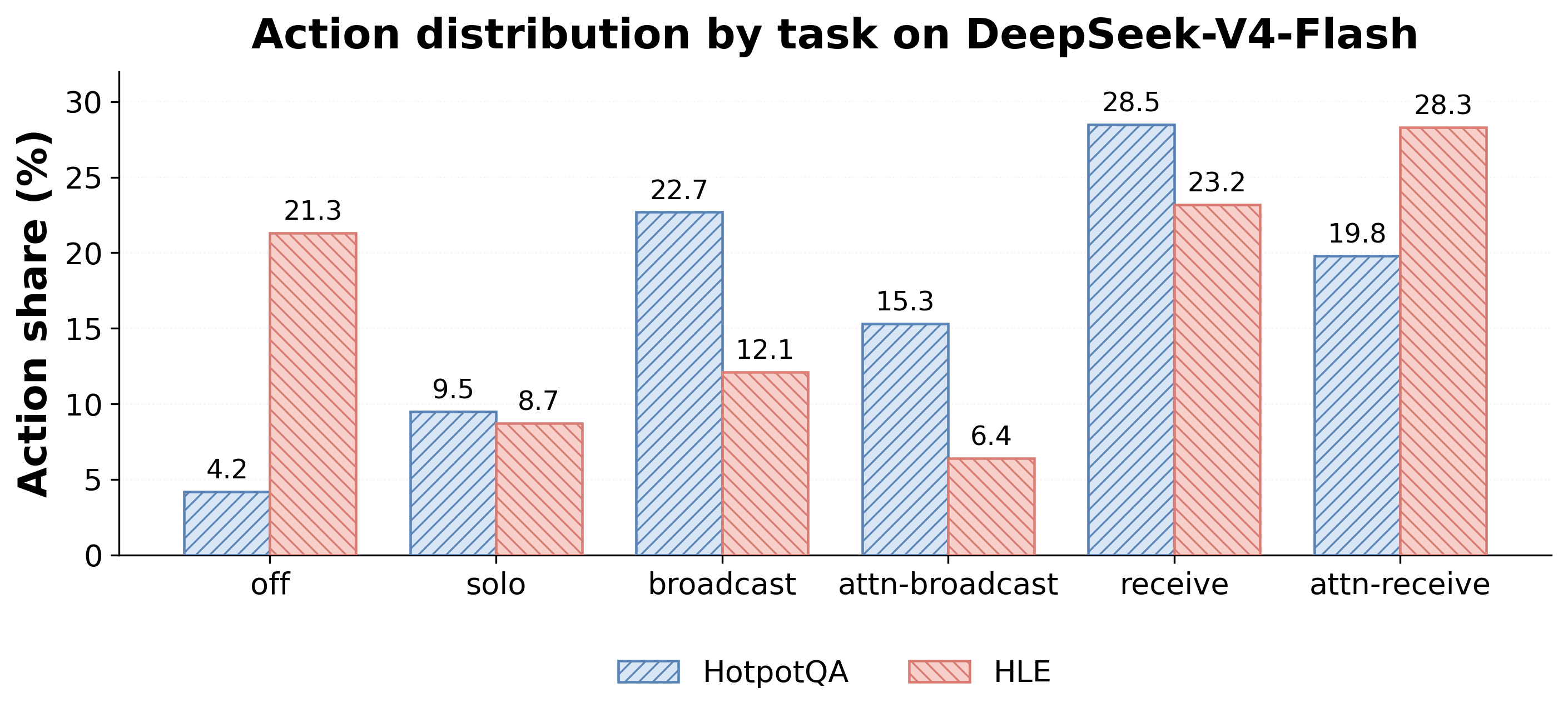}
    \caption{Action distribution by task on DeepSeek-V4-Flash. On HotpotQA,
    the policy sustains agent activity and broadcast communication; on HLE,
    it deactivates a substantial fraction of rounds and relies on pull-based
    retrieval.}
    \label{fig:action_parallel}
\end{wrapfigure}

\texttt{ACTS} achieves its gains through attention-gated selective communication and dynamic deactivation, two complementary mechanisms analyzed on HotpotQA and HLE, two settings that respectively exhibit high and negligible $\mathrm{IG}$. At each round the agent network computes asymmetric directional scores $\alpha_{ij}^t$; \emph{attn-broadcast} retains edge $i\!\to\!j$ when $\alpha_{ij}^t$ exceeds the mean over eligible candidates, \emph{attn-receive} retains $j\!\to\!i$ under the symmetric condition on $\alpha_{ji}^t$, and together these actions construct sparse neighborhoods specific to each round in place of a fixed topology. Trained end to end with the joint accuracy and token objective, the attention module suppresses redundant interactions, as quantified in Table~\ref{tab:token_usage} of Appendix~\ref{sec:appendix_token}. Independently, the \emph{off} action deactivates an agent by suppressing its LLM call and outgoing edges; when a majority of agents become inactive the episode terminates early, forgoing rounds unlikely to alter the answer, whereas agents persist when broader collaboration is warranted, and consequently tasks requiring limited collaboration incur cost approaching that of a single agent. Figure~\ref{fig:action_parallel} contrasts the learned action distributions under DeepSeek-V4-Flash. Under high IG, off stays at $4.2\%$ whereas broadcast reaches $22.7\%$; under IG near zero, off rises to $21.3\%$, broadcast and attn-broadcast drop to $12.1\%$ and $6.4\%$, and attn-receive dominates at $28.3\%$. Appendix~\ref{sec:appendix_case_study} provides a comparison of the two regimes at the trajectory level.

\subsubsection{Ablation Studies}
\label{sec:ablation}

Table~\ref{tab:ablation} isolates the three components distinguishing
\texttt{ACTS} from Agent Q-Mix \citep{agent_qmix}. Replacing the topological attention
mixer with QMIX \citep{qmix} induces the strongest degradation, lowering macro-average
accuracy from $75.01\%$ to $71.83\%$ on GPT-OSS-120B \citep{openai2025gptoss} and from $69.43\%$ to
$63.24\%$ on DeepSeek-V4-Flash \citep{deepseek2026v4flash}, with token usage rising from $5.97$ to
$7.41$ and from $4.03$ to $4.96$. These results support the effectiveness of the proposed topological
attention mixer relative to the QMIX replacement under the evaluated
settings. Substituting fixed role-based
partners for learned relational routing reduces accuracy to $72.50\%$ and
$66.10\%$ and raises token usage to $8.10$ and $6.20$, the largest overhead
among the variants, as fixed assignments cannot adapt partner selection
to the task. With role-specific LLM prompts retained, removing role embeddings
from relational attention reduces macro-average accuracy to
$73.40\%$ and $67.10\%$ and increases token usage to
$6.80\times10^3$ and $4.60\times10^3$ on GPT-OSS-120B and
DeepSeek-V4-Flash, respectively. These results support the contribution
of explicit role conditioning to task performance and communication
efficiency.

\begin{table}[!htbp]
\centering
\caption{Ablation results on the two backbone models.
$\overline{\mathrm{ACC}}$ and $\overline{\mathrm{Token}}$ denote the
macro-average accuracy (\%) and macro-average per-sample token usage
($\times 10^{3}$) across the nine benchmarks. The full
\texttt{ACTS} model is highlighted.}
\label{tab:ablation}

\footnotesize
\setlength{\tabcolsep}{4pt}
\renewcommand{\arraystretch}{1.15}

\begin{tabular}{lrrrr}
\toprule
& \multicolumn{2}{c}{GPT-OSS-120B}
& \multicolumn{2}{c}{DeepSeek-V4-Flash} \\
\cmidrule(lr){2-3} \cmidrule(lr){4-5}
Variant
& $\overline{\mathrm{ACC}}$ & $\overline{\mathrm{Token}}$
& $\overline{\mathrm{ACC}}$ & $\overline{\mathrm{Token}}$ \\
\midrule
\rowcolor{actsblue}
\textbf{ACTS}
& \textbf{75.01} & \textbf{5.97}
& \textbf{69.43} & \textbf{4.03} \\
\midrule
(a) Topological mixer replaced by QMIX
& 71.83 & 7.41 & 63.24 & 4.96 \\
(b) Learned routing replaced by fixed partners
& 72.50 & 8.10 & 66.10 & 6.20 \\
(c) Role embeddings removed
& 73.40 & 6.80 & 67.10 & 4.60 \\
\bottomrule
\end{tabular}
\end{table}

% \begin{table}[!htbp]
% \centering
% \caption{Ablation results on the two backbone models.
% $\overline{\mathrm{ACC}}$ and $\overline{\mathrm{Token}}$ denote the
% macro-average accuracy (\%) and macro-average per-sample token usage
% ($\times 10^{3}$) across the nine benchmarks. The full
% \texttt{ACTS} model is highlighted.}
% \label{tab:ablation}

% \footnotesize
% \setlength{\tabcolsep}{8pt}
% \renewcommand{\arraystretch}{1.15}

% \begin{tabular}{lcccc}
% \toprule
% & \multicolumn{2}{c}{GPT-OSS-120B}
% & \multicolumn{2}{c}{DeepSeek-V4-Flash} \\
% \cmidrule(lr){2-3} \cmidrule(lr){4-5}
% Variant
% & $\overline{\mathrm{ACC}}$ & $\overline{\mathrm{Token}}$
% & $\overline{\mathrm{ACC}}$ & $\overline{\mathrm{Token}}$ \\
% \midrule

% \rowcolor{actsblue}
% \textbf{ACTS}
% & \textbf{75.01} & \textbf{5.97}
% & \textbf{69.43} & \textbf{4.03} \\

% \midrule

% \textit{(a)}\enspace Topological mixer replaced by QMIX
% & 71.83 & 7.41
% & 63.24 & 4.96 \\

% \textit{(b)}\enspace Learned routing replaced by fixed role-based partners
% & 72.50 & 8.10
% & 66.10 & 6.20 \\

% \textit{(c)}\enspace Role embeddings removed from relational attention
% & 73.40 & 6.80
% & 67.10 & 4.60 \\

% \bottomrule
% \end{tabular}
% \end{table}

\section{Conclusion}
\label{sec:conclusion}

We introduced \texttt{ACTS}, a role-aware adaptive communication-topology
scheduling framework for efficient LLM-based multi-agent collaboration.
Under CTDE, its dual-attention design jointly learns when communication
should occur and which directed agent pairs should exchange information:
role-aware relational attention selects variable-cardinality partners
from agent roles and evolving execution states, while a topological
attention mixer incorporates the realized directed topology into team-value
factorization for topology-aware credit assignment. Across nine benchmarks
and two frozen backbone LLMs, \texttt{ACTS} improves the accuracy--efficiency
trade-off over Agent Q-Mix by up to $6.19$ percentage points in accuracy
and up to $15.42\%$ in token reduction. These results demonstrate the
benefit of learning communication modes and directed partner relations
jointly for efficient role-aware LLM collaboration.

\clearpage

\section*{AI Use Statement}

Generative AI tools were used to assist with literature search and reference
verification, drafting and revision of manuscript text, language editing, and
LaTeX formatting. They were not used to collect or generate data, run
experiments, modify empirical results, or determine the final scientific
claims. Other required disclosure categories are not applicable to this work.

All AI-assisted material was reviewed and revised by the authors. The authors
independently checked the manuscript text, equations, proofs, citations,
bibliographic links, reported values, and LaTeX compilation. The final
experimental results and conclusions were verified against the corresponding
experimental records. The authors take full responsibility for the final
content of this work, including all text, claims, references, code, and other
artifacts.

\bibliography{iclr2027_conference}

\begin{thebibliography}{38}
\providecommand{\natexlab}[1]{#1}
\providecommand{\url}[1]{\texttt{#1}}
\expandafter\ifx\csname urlstyle\endcsname\relax
  \providecommand{\doi}[1]{doi: #1}\else
  \providecommand{\doi}{doi: \begingroup \urlstyle{rm}\Url}\fi

\bibitem[{ByteDance-Seed}(2025)]{bytedance2025beyondaime}
{ByteDance-Seed}.
\newblock {BeyondAIME}: Advancing math reasoning evaluation beyond high school olympiads.
\newblock Hugging Face Datasets, 2025.
\newblock URL \url{https://huggingface.co/datasets/ByteDance-Seed/BeyondAIME}.

\bibitem[Chen et~al.(2021)Chen, Tworek, Jun, Yuan, Pinto, Kaplan, Edwards, Burda, Joseph, Brockman, et~al.]{chen2021humaneval}
Mark Chen, Jerry Tworek, Heewoo Jun, Qiming Yuan, Henrique Ponde de~Oliveira Pinto, Jared Kaplan, Harri Edwards, Yuri Burda, Nicholas Joseph, Greg Brockman, et~al.
\newblock Evaluating large language models trained on code.
\newblock \emph{arXiv preprint arXiv:2107.03374}, 2021.
\newblock URL \url{https://arxiv.org/abs/2107.03374}.

\bibitem[Chen et~al.(2023)Chen, Su, Zuo, Yang, Yuan, Chan, Yu, Lu, Hung, Qian, Qin, Cong, Xie, Liu, Sun, and Zhou]{chen2023agentverse}
Weize Chen, Yusheng Su, Jingwei Zuo, Cheng Yang, Chenfei Yuan, Chi-Min Chan, Heyang Yu, Yaxi Lu, Yi-Hsin Hung, Chen Qian, Yujia Qin, Xin Cong, Ruobing Xie, Zhiyuan Liu, Maosong Sun, and Jie Zhou.
\newblock {AgentVerse}: Facilitating multi-agent collaboration and exploring emergent behaviors, 2023.
\newblock URL \url{https://arxiv.org/abs/2308.10848}.

\bibitem[Cover \& Thomas(2006)Cover and Thomas]{cover_thomas}
Thomas~M. Cover and Joy~A. Thomas.
\newblock \emph{Elements of Information Theory}.
\newblock Wiley-Interscience, 2 edition, 2006.
\newblock \doi{10.1002/047174882X}.
\newblock URL \url{https://books.google.com/books?id=j0DBDwAAQBAJ}.

\bibitem[{DeepSeek-AI}(2026)]{deepseek2026v4flash}
{DeepSeek-AI}.
\newblock Deepseek-v4-flash.
\newblock Official model documentation, 2026.
\newblock URL \url{https://api-docs.deepseek.com/}.

\bibitem[Du et~al.(2023)Du, Li, Torralba, Tenenbaum, and Mordatch]{llm_debate}
Yilun Du, Shuang Li, Antonio Torralba, Joshua~B. Tenenbaum, and Igor Mordatch.
\newblock Improving factuality and reasoning in language models through multiagent debate, 2023.
\newblock URL \url{https://arxiv.org/abs/2305.14325}.

\bibitem[Foerster et~al.(2016)Foerster, Assael, de~Freitas, and Whiteson]{foerster2016}
Jakob~N. Foerster, Ioannis~Alexandros Assael, Nando de~Freitas, and Shimon Whiteson.
\newblock Learning to communicate with deep multi-agent reinforcement learning.
\newblock In \emph{Advances in Neural Information Processing Systems}, volume~29, 2016.
\newblock URL \url{https://papers.nips.cc/paper_files/paper/2016/hash/c7635bfd99248a2cdef8249ef7bfbef4-Abstract.html}.

\bibitem[Hong et~al.(2023)Hong, Zhuge, Chen, Zheng, Cheng, Zhang, Wang, Wang, Yau, Lin, Zhou, Ran, Xiao, Wu, and Schmidhuber]{hong2023metagpt}
Sirui Hong, Mingchen Zhuge, Jonathan Chen, Xiawu Zheng, Yuheng Cheng, Ceyao Zhang, Jinlin Wang, Zili Wang, Steven Ka~Shing Yau, Zijuan Lin, Liyang Zhou, Chenyu Ran, Lingfeng Xiao, Chenglin Wu, and J{\"u}rgen Schmidhuber.
\newblock {MetaGPT}: Meta programming for a multi-agent collaborative framework, 2023.
\newblock URL \url{https://arxiv.org/abs/2308.00352v5}.

\bibitem[Jain et~al.(2024)Jain, Han, Gu, Li, Yan, Zhang, Wang, Solar-Lezama, Sen, and Stoica]{jain2024livecodebench}
Naman Jain, King Han, Alex Gu, Wen-Ding Li, Fanjia Yan, Tianjun Zhang, Sida Wang, Armando Solar-Lezama, Koushik Sen, and Ion Stoica.
\newblock {LiveCodeBench}: Holistic and contamination free evaluation of large language models for code.
\newblock \emph{arXiv preprint arXiv:2403.07974}, 2024.
\newblock URL \url{https://arxiv.org/abs/2403.07974}.

\bibitem[Jiang et~al.(2025)Jiang, Li, Wan, Yin, Wu, Liang, Li, Sun, Wang, Chang, and Wu]{gtd}
Eric~Hanchen Jiang, Mengting Li, Guancheng Wan, Sophia Yin, Yuchen Wu, Xiao Liang, Xinfeng Li, Yizhou Sun, Wei Wang, Kai-Wei Chang, and Ying~Nian Wu.
\newblock Dynamic generation of multi-{LLM} agents communication topologies with graph diffusion models, 2025.
\newblock URL \url{https://arxiv.org/abs/2510.07799}.

\bibitem[Jiang et~al.(2026)Jiang, Li, Sun, Liang, Li, Wu, Luo, Li, Zhang, Kang, Chang, and Wu]{agent_qmix}
Eric~Hanchen Jiang, Levina Li, Rui Sun, Xiao Liang, Yubei Li, Yuchen Wu, Haozheng Luo, Hengli Li, Zhi Zhang, Zhaolu Kang, Kai-Wei Chang, and Ying~Nian Wu.
\newblock {Agent Q-Mix}: Selecting the right action for {LLM} multi-agent systems through reinforcement learning, 2026.
\newblock URL \url{https://arxiv.org/abs/2604.00344}.

\bibitem[{LangChain AI}(2024)]{langchain2024langgraph}
{LangChain AI}.
\newblock {LangGraph}.
\newblock Software repository, 2024.
\newblock URL \url{https://github.com/langchain-ai/langgraph}.

\bibitem[Li et~al.(2023)Li, Hammoud, Itani, Khizbullin, and Ghanem]{li2023camel}
Guohao Li, Hasan Hammoud, Hani Itani, Dmitrii Khizbullin, and Bernard Ghanem.
\newblock {CAMEL}: Communicative agents for "mind" exploration of large language model society.
\newblock In \emph{Advances in Neural Information Processing Systems}, volume~36, 2023.
\newblock URL \url{https://proceedings.neurips.cc/paper_files/paper/2023/hash/a3621ee907def47c1b952ade25c67698-Abstract-Conference.html}.

\bibitem[Li et~al.(2025)Li, Liu, Wen, Zhang, and Pan]{li2025assemble}
Shiyuan Li, Yixin Liu, Qingsong Wen, Chengqi Zhang, and Shirui Pan.
\newblock Assemble your crew: Automatic multi-agent communication topology design via autoregressive graph generation, 2025.
\newblock URL \url{https://arxiv.org/abs/2507.18224}.

\bibitem[Liu et~al.(2024)Liu, Zhang, Li, Liu, and Yang]{liu2024dylan}
Zijun Liu, Yanzhe Zhang, Peng Li, Yang Liu, and Diyi Yang.
\newblock A dynamic {LLM}-powered agent network for task-oriented agent collaboration.
\newblock In \emph{Proceedings of the First Conference on Language Modeling}, 2024.
\newblock URL \url{https://openreview.net/forum?id=XII0Wp1XA9}.

\bibitem[{MathArena}(2025)]{matharena2025hmmt}
{MathArena}.
\newblock {HMMT} february 2025 competition problems.
\newblock Hugging Face Datasets, 2025.
\newblock URL \url{https://huggingface.co/datasets/MathArena/hmmt_feb_2025}.

\bibitem[{Mathematical Association of America}(2025)]{maa2025aime}
{Mathematical Association of America}.
\newblock {2025 AIME I}.
\newblock AoPS Wiki, 2025.
\newblock URL \url{https://artofproblemsolving.com/wiki/index.php/2025_AIME_I}.

\bibitem[{Mathematical Association of America}(2026)]{maa2026aime}
{Mathematical Association of America}.
\newblock {2026 AIME I}.
\newblock AoPS Wiki, 2026.
\newblock URL \url{https://artofproblemsolving.com/wiki/index.php/2026_AIME_I}.

\bibitem[{Microsoft}(2025)]{microsoft2025agentframework}
{Microsoft}.
\newblock {Microsoft Agent Framework}.
\newblock Software repository, 2025.
\newblock URL \url{https://github.com/microsoft/agent-framework}.

\bibitem[Oliehoek \& Amato(2016)Oliehoek and Amato]{oliehoek2016}
Frans~A. Oliehoek and Christopher Amato.
\newblock \emph{A Concise Introduction to Decentralized {POMDPs}}.
\newblock Springer, 2016.
\newblock \doi{10.1007/978-3-319-28929-8}.
\newblock URL \url{https://link.springer.com/book/10.1007/978-3-319-28929-8}.

\bibitem[{OpenAI}(2025)]{openai2025gptoss}
{OpenAI}.
\newblock {gpt-oss-120b} \& {gpt-oss-20b} model card.
\newblock \emph{arXiv preprint arXiv:2508.10925}, 2025.
\newblock URL \url{https://arxiv.org/abs/2508.10925}.

\bibitem[Phan et~al.(2025)Phan, Gatti, Han, Li, Hu, Zhang, Zhang, Shaaban, Ling, Shi, et~al.]{hle2025}
Long Phan, Alice Gatti, Ziwen Han, Nathaniel Li, Josephina Hu, Hugh Zhang, Chen Bo~Calvin Zhang, Mohamed Shaaban, John Ling, Sean Shi, et~al.
\newblock Humanity's last exam.
\newblock \emph{arXiv preprint arXiv:2501.14249}, 2025.
\newblock URL \url{https://arxiv.org/abs/2501.14249}.

\bibitem[Qian et~al.(2024)Qian, Liu, Liu, Chen, Dang, Li, Yang, Chen, Su, Cong, Xu, Li, Liu, and Sun]{qian2024chatdev}
Chen Qian, Wei Liu, Hongzhang Liu, Nuo Chen, Yufan Dang, Jiahao Li, Cheng Yang, Weize Chen, Yusheng Su, Xin Cong, Juyuan Xu, Dahai Li, Zhiyuan Liu, and Maosong Sun.
\newblock {ChatDev}: Communicative agents for software development.
\newblock In \emph{Proceedings of the 62nd Annual Meeting of the Association for Computational Linguistics (Volume 1: Long Papers)}, pp.\  15174--15186. Association for Computational Linguistics, 2024.
\newblock \doi{10.18653/v1/2024.acl-long.810}.
\newblock URL \url{https://aclanthology.org/2024.acl-long.810/}.

\bibitem[Rashid et~al.(2018)Rashid, Samvelyan, Schroeder~de Witt, Farquhar, Foerster, and Whiteson]{qmix}
Tabish Rashid, Mikayel Samvelyan, Christian Schroeder~de Witt, Gregory Farquhar, Jakob~N. Foerster, and Shimon Whiteson.
\newblock {QMIX}: Monotonic value function factorisation for deep multi-agent reinforcement learning.
\newblock In \emph{Proceedings of the 35th International Conference on Machine Learning}, volume~80 of \emph{Proceedings of Machine Learning Research}, pp.\  4295--4304. PMLR, 2018.
\newblock URL \url{https://proceedings.mlr.press/v80/rashid18a.html}.

\bibitem[Son et~al.(2019)Son, Kim, Kang, Hostallero, and Yi]{son2019qtran}
Kyunghwan Son, Daewoo Kim, Wan~Ju Kang, David~Earl Hostallero, and Yung Yi.
\newblock {QTRAN}: Learning to factorize with transformation for cooperative multi-agent reinforcement learning.
\newblock In \emph{Proceedings of the 36th International Conference on Machine Learning}, volume~97 of \emph{Proceedings of Machine Learning Research}, pp.\  5887--5896. PMLR, 2019.
\newblock URL \url{https://proceedings.mlr.press/v97/son19a.html}.

\bibitem[Sukhbaatar et~al.(2016)Sukhbaatar, Szlam, and Fergus]{sukhbaatar2016}
Sainbayar Sukhbaatar, Arthur Szlam, and Rob Fergus.
\newblock Learning multiagent communication with backpropagation.
\newblock In \emph{Advances in Neural Information Processing Systems}, volume~29, 2016.
\newblock URL \url{https://papers.nips.cc/paper_files/paper/2016/hash/55b1927fdafef39c48e5b73b5d61ea60-Abstract.html}.

\bibitem[Sun et~al.(2026)Sun, Ding, Gong, Gu, Jiang, Zhang, Pan, and L{\"u}]{topodim}
Rui Sun, Jie Ding, Chenghua Gong, Tianjun Gu, Yihang Jiang, Juyuan Zhang, Liming Pan, and Linyuan L{\"u}.
\newblock {TopoDIM}: One-shot topology generation of diverse interaction modes for multi-agent systems, 2026.
\newblock URL \url{https://arxiv.org/abs/2601.10120}.

\bibitem[Sunehag et~al.(2018)Sunehag, Lever, Gruslys, Czarnecki, Zambaldi, Jaderberg, Lanctot, Sonnerat, Leibo, Tuyls, and Graepel]{sunehag2018}
Peter Sunehag, Guy Lever, Audrunas Gruslys, Wojciech~Marian Czarnecki, Vinicius Zambaldi, Max Jaderberg, Marc Lanctot, Nicolas Sonnerat, Joel~Z. Leibo, Karl Tuyls, and Thore Graepel.
\newblock Value-decomposition networks for cooperative multi-agent learning based on team reward.
\newblock In \emph{Proceedings of the 17th International Conference on Autonomous Agents and MultiAgent Systems}, pp.\  2085--2087, 2018.
\newblock URL \url{https://www.ifaamas.org/Proceedings/aamas2018/pdfs/p2085.pdf}.

\bibitem[Wang et~al.(2021)Wang, Ren, Liu, Yu, and Zhang]{wang2021qplex}
Jianhao Wang, Zhizhou Ren, Terry Liu, Yang Yu, and Chongjie Zhang.
\newblock {QPLEX}: Duplex dueling multi-agent q-learning.
\newblock In \emph{International Conference on Learning Representations}, 2021.
\newblock URL \url{https://openreview.net/forum?id=Rcmk0xxIQV}.

\bibitem[Wang et~al.(2024)Wang, Ma, Zhang, Ni, Chandra, Guo, Ren, Arulraj, He, Jiang, Li, Ku, Wang, Zhuang, Fan, Yue, and Chen]{wang2024mmlupro}
Yubo Wang, Xueguang Ma, Ge~Zhang, Yuansheng Ni, Abhranil Chandra, Shiguang Guo, Weiming Ren, Aaran Arulraj, Xuan He, Ziyan Jiang, Tianle Li, Max Ku, Kai Wang, Alex Zhuang, Rongqi Fan, Xiang Yue, and Wenhu Chen.
\newblock {MMLU-Pro}: A more robust and challenging multi-task language understanding benchmark.
\newblock In \emph{Advances in Neural Information Processing Systems}, volume~37, 2024.
\newblock \doi{10.52202/079017-3018}.
\newblock URL \url{https://proceedings.neurips.cc/paper_files/paper/2024/hash/ad236edc564f3e3156e1b2feafb99a24-Abstract.html}.

\bibitem[Wang et~al.(2025)Wang, Wang, Liu, Ding, Zhang, Liu, and Zhang]{wang2025agentdropout}
Zhexuan Wang, Yutong Wang, Xuebo Liu, Liang Ding, Miao Zhang, Jie Liu, and Min Zhang.
\newblock {AgentDropout}: Dynamic agent elimination for token-efficient and high-performance {LLM}-based multi-agent collaboration, 2025.
\newblock URL \url{https://arxiv.org/abs/2503.18891}.

\bibitem[Wu et~al.(2023)Wu, Bansal, Zhang, Wu, Li, Zhu, Jiang, Zhang, Zhang, Liu, Awadallah, White, Burger, and Wang]{autogen}
Qingyun Wu, Gagan Bansal, Jieyu Zhang, Yiran Wu, Beibin Li, Erkang Zhu, Li~Jiang, Xiaoyun Zhang, Shaokun Zhang, Jiale Liu, Ahmed~Hassan Awadallah, Ryen~W. White, Doug Burger, and Chi Wang.
\newblock {AutoGen}: Enabling next-gen {LLM} applications via multi-agent conversation, 2023.
\newblock URL \url{https://arxiv.org/abs/2308.08155}.

\bibitem[Yan et~al.(2025)Yan, Zhou, Zhang, Zhang, Zhou, Miao, Li, Li, and Zhang]{yan2025beyondselftalk}
Bingyu Yan, Zhibo Zhou, Litian Zhang, Lian Zhang, Ziyi Zhou, Dezhuang Miao, Zhoujun Li, Chaozhuo Li, and Xiaoming Zhang.
\newblock Beyond self-talk: A communication-centric survey of llm-based multi-agent systems.
\newblock \emph{arXiv preprint arXiv:2502.14321}, 2025.
\newblock \doi{10.48550/arXiv.2502.14321}.
\newblock URL \url{https://arxiv.org/abs/2502.14321}.

\bibitem[Yang et~al.(2020)Yang, Hao, Liao, Shao, Chen, Liu, and Tang]{yang2020qatten}
Yaodong Yang, Jianye Hao, Ben Liao, Kun Shao, Guangyong Chen, Wulong Liu, and Hongyao Tang.
\newblock {Qatten}: A general framework for cooperative multiagent reinforcement learning.
\newblock \emph{arXiv preprint arXiv:2002.03939}, 2020.
\newblock URL \url{https://arxiv.org/abs/2002.03939}.

\bibitem[Yang et~al.(2018)Yang, Qi, Zhang, Bengio, Cohen, Salakhutdinov, and Manning]{yang2018hotpotqa}
Zhilin Yang, Peng Qi, Saizheng Zhang, Yoshua Bengio, William~W. Cohen, Ruslan Salakhutdinov, and Christopher~D. Manning.
\newblock {HotpotQA}: A dataset for diverse, explainable multi-hop question answering.
\newblock In \emph{Proceedings of the 2018 Conference on Empirical Methods in Natural Language Processing}, pp.\  2369--2380. Association for Computational Linguistics, 2018.
\newblock \doi{10.18653/v1/D18-1259}.
\newblock URL \url{https://aclanthology.org/D18-1259/}.

\bibitem[Zhang et~al.(2024)Zhang, Yue, Sun, Wan, Yu, Fang, Wang, Chen, and Cheng]{gdesigner}
Guibin Zhang, Yanwei Yue, Xiangguo Sun, Guancheng Wan, Miao Yu, Junfeng Fang, Kun Wang, Tianlong Chen, and Dawei Cheng.
\newblock {G-Designer}: Architecting multi-agent communication topologies via graph neural networks, 2024.
\newblock URL \url{https://arxiv.org/abs/2410.11782}.

\bibitem[Zhang et~al.(2025)Zhang, Niu, Fang, Wang, Bai, and Wang]{maas}
Guibin Zhang, Luyang Niu, Junfeng Fang, Kun Wang, Lei Bai, and Xiang Wang.
\newblock Multi-agent architecture search via agentic supernet, 2025.
\newblock URL \url{https://arxiv.org/abs/2502.04180}.

\bibitem[Zhuge et~al.(2024)Zhuge, Wang, Kirsch, Faccio, Khizbullin, and Schmidhuber]{zhuge2024gptswarm}
Mingchen Zhuge, Wenyi Wang, Louis Kirsch, Francesco Faccio, Dmitrii Khizbullin, and J{\"u}rgen Schmidhuber.
\newblock {GPTSwarm}: Language agents as optimizable graphs.
\newblock In \emph{Proceedings of the 41st International Conference on Machine Learning}, volume 235 of \emph{Proceedings of Machine Learning Research}, pp.\  62743--62767. PMLR, 2024.
\newblock URL \url{https://proceedings.mlr.press/v235/zhuge24a.html}.

\end{thebibliography}
\bibliographystyle{iclr2027_conference}

\appendix
\section{Value Function Factorization in Cooperative Multi-Agent Reinforcement Learning}
\label{app:marl_value_factorization}

\subsection{Centralized Training with Decentralized Execution}

Cooperative multi-agent reinforcement learning (MARL) considers a set of
agents that jointly maximize a shared return. Let $s^t$ denote the
global state, $\tau_i^t$ the local action-observation history of agent
$i$, and
$\boldsymbol u^t=(u_1^t,\ldots,u_N^t)$ the joint action. The objective is
\begin{equation}
J(\pi)
=
\mathbb E_{\pi}
\left[
\sum_{t=1}^{T}
\gamma^{t-1}r^t
\right],
\label{eq:app_marl_return}
\end{equation}
where $r^t$ is the shared reward and $\gamma\in[0,1)$ is the discount
factor.

In partially observable environments, each agent makes decisions from
its local action-observation history. Independent value-based learning is
challenging because the joint action space grows exponentially with the
number of agents, while the effective environment of each agent changes
as the other agents update their policies. Centralized training with
decentralized execution (CTDE) addresses these issues by allowing the
training procedure to access centralized information, while restricting
execution to local histories.

Let $Q_i(\tau_i^t,u_i)$ denote the individual action-value function of
agent $i$. Under CTDE, the centralized joint action-value function can
be represented as
\begin{equation}
Q_{\mathrm{tot}}
\left(
s^t,\boldsymbol\tau^t,\boldsymbol u^t
\right)
=
F_{\psi}
\left(
s^t,
\left\{
Q_i(\tau_i^t,u_i^t)
\right\}_{i=1}^{N}
\right),
\label{eq:app_value_factorization}
\end{equation}
where $F_{\psi}$ is a centralized factorization function. This
formulation allows centralized training to model inter-agent
dependencies while retaining decentralized action selection at
execution time.

\subsection{Value Function Factorization}

Value Decomposition Networks (VDN) use the additive factorization
\begin{equation}
Q_{\mathrm{tot}}
\left(
\boldsymbol\tau,\boldsymbol u
\right)
=
\sum_{i=1}^{N}
Q_i(\tau_i,u_i),
\label{eq:app_vdn}
\end{equation}
which supports decentralized execution but cannot represent general
non-additive interactions among agents \citep{sunehag2018}.

QMIX introduces a nonlinear factorization function
\citep{qmix}:
\begin{equation}
Q_{\mathrm{tot}}
=
F_{\psi}
\left(
Q_1,\ldots,Q_N;s
\right),
\label{eq:app_qmix}
\end{equation}
subject to the monotonicity condition
\begin{equation}
\frac{\partial Q_{\mathrm{tot}}}{\partial Q_i}
\geq 0,
\qquad
i=1,\ldots,N.
\label{eq:app_monotonicity}
\end{equation}
Thus, increasing an individual utility cannot decrease the joint value
when the remaining utilities and the centralized state are fixed. QMIX
typically enforces this condition by generating nonnegative factorization
weights from the global state through hypernetworks.

Subsequent methods improve the expressiveness of value function
factorization in different ways. QTRAN introduces transformation and
consistency constraints to represent a broader class of joint value
functions \citep{son2019qtran}. QPLEX uses a duplex dueling architecture
to construct a more expressive value--advantage factorization while
preserving decentralized consistency \citep{wang2021qplex}. Attention-
based factorization methods assign context-dependent weights to
individual utilities, allowing their contributions to vary with the
centralized context \citep{yang2020qatten}.

\subsection{Monotonic Value Mixing and the IGM Property}

A factorized joint action-value function satisfies the
Individual--Global--Max (IGM) \citep{sunehag2018, qmix, son2019qtran} property if its globally greedy action is
consistent with the individual greedy actions:
\begin{equation}
\operatorname*{arg\,max}_{\boldsymbol u}
Q_{\mathrm{tot}}
\left(
s,\boldsymbol\tau,\boldsymbol u
\right)
=
\left(
\operatorname*{arg\,max}_{u_1}Q_1(\tau_1,u_1),
\ldots,
\operatorname*{arg\,max}_{u_N}Q_N(\tau_N,u_N)
\right).
\label{eq:app_igm_definition}
\end{equation}
With consistent tie-breaking, this property means that the globally
optimal joint action can be recovered by independently maximizing the
individual action-value functions.

Componentwise monotonicity is a sufficient condition for IGM. For fixed
$s$ and $\boldsymbol\tau$, consider two joint actions that differ only
in agent $i$'s action. If one action produces a larger value of
$Q_i(\tau_i,u_i)$, the monotonicity condition ensures that it cannot
produce a smaller value of $Q_{\mathrm{tot}}$. Applying this argument
coordinate-wise to all agents shows that independent greedy action
selection yields a globally optimal joint action.

\subsection{Temporal-Difference Optimization}

Cooperative value function factorization methods are commonly trained with a
centralized temporal-difference objective while retaining decentralized
execution \citep{sunehag2018,qmix,son2019qtran}. For a transition
$(s^t,\boldsymbol\tau^t,\boldsymbol u^t,r^t,s^{t+1},d^t)$, where
$d^t$ indicates episode termination, a target network forms
\begin{equation}
y^t
=
r^t+
\gamma(1-d^t)
Q_{\mathrm{tot}}^{-}
\left(
s^{t+1},
\boldsymbol\tau^{t+1},
\boldsymbol u^{*,t+1}
\right).
\label{eq:app_td_target}
\end{equation}

When the factorization satisfies the Individual--Global-Max property, the
target joint action can be selected independently
\citep{qmix,wang2021qplex}:
\begin{equation}
u_i^{*,t+1}
=
\operatorname*{arg\,max}_{u_i}
Q_i^{-}
\left(
\tau_i^{t+1},u_i
\right),
\qquad
i=1,\ldots,N.
\label{eq:app_target_action}
\end{equation}

The online parameters are optimized by minimizing
\begin{equation}
\mathcal L_{\mathrm{TD}}
=
\mathbb E_{\mathcal D}
\left[
\left(
y^t-
Q_{\mathrm{tot}}
\left(
s^t,\boldsymbol\tau^t,\boldsymbol u^t
\right)
\right)^2
\right].
\label{eq:app_td_loss}
\end{equation}

Experience replay and delayed target networks are commonly used to stabilize
optimization under deep function approximation
\citep{qmix,son2019qtran}.

\section{Construction of the Directed Communication Topology}
\label{app:directed_topology_construction}

This section gives the complete construction of
$G^{t+1}=\Psi(\boldsymbol a^t,\boldsymbol\alpha^t)$.
After communication-mode selection, define the communication-eligible
agents and candidate partners as
\begin{equation}
\mathcal I_{t+1}^{\mathrm{comm}}
=
\left\{
i\in\mathcal I
\;\middle|\;
a_i^t
\notin
\left\{
a^{\mathrm{off}},a^{\mathrm{solo}}
\right\}
\right\},
\qquad
\mathcal J_i^{t+1}
=
\mathcal I_{t+1}^{\mathrm{comm}}\setminus\{i\}.
\label{eq:app_communication_sets}
\end{equation}
The index $t+1$ indicates that these sets determine the topology
produced by the decisions at round $t$.

For
$i\in\mathcal I_{t+1}^{\mathrm{comm}}$
with
$\mathcal J_i^{t+1}\neq\varnothing$,
define the outgoing and incoming mean scores as
\begin{equation}
\bar{\alpha}_i^{\mathrm{out},t}
=
\frac{1}{|\mathcal J_i^{t+1}|}
\sum_{j\in\mathcal J_i^{t+1}}
\alpha_{ij}^t,
\qquad
\bar{\alpha}_i^{\mathrm{in},t}
=
\frac{1}{|\mathcal J_i^{t+1}|}
\sum_{j\in\mathcal J_i^{t+1}}
\alpha_{ji}^t.
\label{eq:app_mean_relational_scores}
\end{equation}

The requested receiver and sender sets are
\begin{equation}
\begin{aligned}
\mathcal S_i^{t+1}
&=
\begin{cases}
\mathcal J_i^{t+1},
& a_i^t=a^{\mathrm{broadcast}},\\
\left\{
j\in\mathcal J_i^{t+1}
\;\middle|\;
\alpha_{ij}^t>
\bar{\alpha}_i^{\mathrm{out},t}
\right\},
& a_i^t=a^{\mathrm{attn\text{-}broadcast}},\\
\varnothing,
& \text{otherwise},
\end{cases}
\\[1mm]
\mathcal R_i^{t+1}
&=
\begin{cases}
\mathcal J_i^{t+1},
& a_i^t=a^{\mathrm{receive}},\\
\left\{
j\in\mathcal J_i^{t+1}
\;\middle|\;
\alpha_{ji}^t>
\bar{\alpha}_i^{\mathrm{in},t}
\right\},
& a_i^t=a^{\mathrm{attn\text{-}receive}},\\
\varnothing,
& \text{otherwise}.
\end{cases}
\end{aligned}
\label{eq:app_partner_request_sets}
\end{equation}
When
$\mathcal J_i^{t+1}=\varnothing$,
both request sets are defined as empty.

The directed edge set is then
\begin{equation}
\mathcal E^{t+1}
=
\left\{
(i,j)\in
\mathcal I_{t+1}^{\mathrm{comm}}
\times
\mathcal I_{t+1}^{\mathrm{comm}}
\;\middle|\;
i\neq j,\;
j\in\mathcal S_i^{t+1}
\lor
i\in\mathcal R_j^{t+1}
\right\},
\label{eq:app_directed_edge_set}
\end{equation}
and
\begin{equation}
G^{t+1}
=
\left(
\mathcal I,\mathcal E^{t+1}
\right).
\label{eq:app_next_topology}
\end{equation}

Thus, $a^{\mathrm{broadcast}}$ and $a^{\mathrm{receive}}$ request
all eligible interactions in the corresponding direction, whereas
$a^{\mathrm{attn\text{-}broadcast}}$ and
$a^{\mathrm{attn\text{-}receive}}$ retain only interactions whose
scores exceed the corresponding mean. The strict mean threshold
allows the number of selected partners to adapt to the score
distribution without introducing a fixed top-$k$ hyperparameter.

\section{Theoretical Analysis of Topological Value Factorization}
\label{app:acts_theory}

At a non-terminal round $t$, let
$\mathcal I_{\mathrm{comm}}^t\neq\varnothing$ denote the set of agents
eligible to make communication decisions, fixed before the current actions
are selected. The corresponding pre-decision context is
\begin{equation}
\mathcal C^t
=
\left(
s^t,\;G^t,\;\mathcal I_{\mathrm{comm}}^t,\;
\{\xi_i^t,z_i^t\}_{i\in\mathcal I_{\mathrm{comm}}^t}
\right),
\label{eq:app_context}
\end{equation}
where $\xi_i^t$ is the input to the individual action-value function and
$z_i^t$ encodes agent $i$'s directed neighborhoods in $G^t$ for centralized
value factorization. Let $\boldsymbol a^t=(a_i^t)_{i\in\mathcal I}$ denote
the full joint action, with $a_i^t=a^{\mathrm{off}}$ for every
$i\notin\mathcal I_{\mathrm{comm}}^t$. The selected actions determine
\begin{equation}
\mathcal I_{\mathrm{comm}}^{t+1}
=
\left\{
i\in\mathcal I:\;
a_i^t\notin\{a^{\mathrm{off}},a^{\mathrm{solo}}\}
\right\},
\qquad
G^{t+1}
=
\Psi\!\left(\mathcal C^t,\boldsymbol a^t\right),
\label{eq:app_topology_transition}
\end{equation}
where $\Psi$ denotes the construction of
Appendix~\ref{app:directed_topology_construction}. Because every agent
outside $\mathcal I_{\mathrm{comm}}^t$ is assigned $a^{\mathrm{off}}$,
Equation~\ref{eq:app_topology_transition} implies
$\mathcal I_{\mathrm{comm}}^{t+1}\subseteq\mathcal I_{\mathrm{comm}}^t$: the
eligible set contracts monotonically along an episode, so an agent that
selects $a^{\mathrm{off}}$ or $a^{\mathrm{solo}}$ withdraws from
communication for the remainder of the episode.

\subsection{Local Expansion and Interaction-Dependent Coefficients}
\label{app:local_value_expansion}

Fix $\mathcal C^t$ and let
$\boldsymbol q=(q_i)_{i\in\mathcal I_{\mathrm{comm}}^t}$ denote a vector of
individual utility values. Suppose a reference joint action-value function is
represented locally by a map
\[
F_{\mathcal C^t}:\mathcal U\rightarrow\mathbb R,
\qquad
\mathcal U\subseteq
\mathbb R^{|\mathcal I_{\mathrm{comm}}^t|}
\ \text{open},
\qquad
F_{\mathcal C^t}\in C^2(\mathcal U).
\]
This smoothness assumption is imposed in utility space and therefore does not
require the discrete communication actions to be differentiable. For
$\boldsymbol u\in\mathcal U$ define the local sensitivity
\begin{equation}
\beta_i^{\mathrm{ref},t}(\boldsymbol u)
=
\partial_iF_{\mathcal C^t}(\boldsymbol u),
\label{eq:app_local_sensitivity}
\end{equation}
and, for a reference point $\bar{\boldsymbol q}\in\mathcal U$, write
\begin{equation}
\beta_i^{\mathrm{ref},t}
=
\beta_i^{\mathrm{ref},t}(\bar{\boldsymbol q}),
\qquad
c^{\mathrm{ref},t}
=
F_{\mathcal C^t}(\bar{\boldsymbol q})
-
\sum_{i\in\mathcal I_{\mathrm{comm}}^t}
\beta_i^{\mathrm{ref},t}\bar q_i .
\label{eq:app_reference_coefficients}
\end{equation}

\begin{proposition}[Local weighted expansion]
\label{prop:acts_local_expansion}
If the line segment between $\bar{\boldsymbol q}$ and $\boldsymbol q$ lies in
$\mathcal U$, then
\begin{equation}
F_{\mathcal C^t}(\boldsymbol q)
=
c^{\mathrm{ref},t}
+
\sum_{i\in\mathcal I_{\mathrm{comm}}^t}
\beta_i^{\mathrm{ref},t}q_i
+
R_2^t .
\label{eq:app_local_weighted_form}
\end{equation}
If in addition
$\|\nabla^2F_{\mathcal C^t}(\boldsymbol u)\|_{2\rightarrow2}\leq L$
along that segment, then
\begin{equation}
\left|R_2^t\right|
\leq
\frac{L}{2}
\left\|\boldsymbol q-\bar{\boldsymbol q}\right\|_2^2 .
\label{eq:app_local_remainder}
\end{equation}
\end{proposition}

\begin{proof}
Let $\boldsymbol\Delta=\boldsymbol q-\bar{\boldsymbol q}$ and
$g(u)=F_{\mathcal C^t}(\bar{\boldsymbol q}+u\boldsymbol\Delta)$. Since
$g\in C^2([0,1])$, the integral form of Taylor's theorem gives
\begin{equation}
\begin{aligned}
F_{\mathcal C^t}(\boldsymbol q)
={}&
F_{\mathcal C^t}(\bar{\boldsymbol q})
+
\nabla F_{\mathcal C^t}
(\bar{\boldsymbol q})^\top\boldsymbol\Delta
\\
&+
\int_0^1(1-u)\,
\boldsymbol\Delta^\top
\nabla^2F_{\mathcal C^t}
(\bar{\boldsymbol q}+u\boldsymbol\Delta)
\boldsymbol\Delta\,du .
\end{aligned}
\label{eq:app_taylor_expansion}
\end{equation}
Because $\Delta_i=q_i-\bar q_i$ and
$\beta_i^{\mathrm{ref},t}=\partial_iF_{\mathcal C^t}(\bar{\boldsymbol q})$,
the first two terms of Equation~\ref{eq:app_taylor_expansion} equal
$c^{\mathrm{ref},t}+\sum_i\beta_i^{\mathrm{ref},t}q_i$ by
Equation~\ref{eq:app_reference_coefficients}, which yields
Equation~\ref{eq:app_local_weighted_form} with $R_2^t$ the integral term. For
the remainder,
$|\boldsymbol\Delta^\top\nabla^2F_{\mathcal C^t}\boldsymbol\Delta|
\leq
\|\nabla^2F_{\mathcal C^t}\|_{2\rightarrow2}
\|\boldsymbol\Delta\|_2^2
\leq L\|\boldsymbol\Delta\|_2^2$,
and $\int_0^1(1-u)\,du=\tfrac12$ gives
Equation~\ref{eq:app_local_remainder}.
\end{proof}

\begin{remark}
\label{rem:acts_expansion_scope}
Proposition~\ref{prop:acts_local_expansion} is a statement about the
\emph{form} of a local approximation rather than a claim of novelty: every
$C^2$ reference value admits weighted-sum coefficients equal to its local
sensitivities. Its role here is to identify what a mixer must supply, namely
per-agent coefficients that vary with the pre-decision context, and to bound
the error of the weighted-sum form by the curvature of $F_{\mathcal C^t}$.
The analysis that follows concerns \emph{which} quantities those coefficients
depend on.
\end{remark}

The coefficient $\beta_i^{\mathrm{ref},t}$ is the local sensitivity of the
reference joint value to agent $i$'s utility, and its variation across
operating points is governed by mixed derivatives. Applying the fundamental
theorem of calculus to
$u\mapsto\partial_iF_{\mathcal C^t}(\bar{\boldsymbol q}+u\boldsymbol\delta)$
and using the chain rule gives, for any displacement $\boldsymbol\delta$
whose segment lies in $\mathcal U$,
\begin{equation}
\beta_i^{\mathrm{ref},t}
(\bar{\boldsymbol q}+\boldsymbol\delta)
-
\beta_i^{\mathrm{ref},t}(\bar{\boldsymbol q})
=
\sum_{j\in\mathcal I_{\mathrm{comm}}^t}
\delta_j
\int_0^1
\partial_{ij}F_{\mathcal C^t}
(\bar{\boldsymbol q}+u\boldsymbol\delta)\,du .
\label{eq:app_coefficient_variation}
\end{equation}
For $j\neq i$, the mixed derivative
$\partial_{ij}F_{\mathcal C^t}$ characterizes how agent $j$'s utility
affects the joint value's local sensitivity to agent $i$'s utility.
This dependence motivates context-dependent mixing coefficients,
consistent with attention-based value factorization
\citep{yang2020qatten}.

A direct connection to the communication neighborhoods follows once the
reference value is graph structured. Let
\begin{equation}
\mathcal E_{\mathrm{comm}}^t
=
\mathcal E^t\cap
\left(
\mathcal I_{\mathrm{comm}}^t
\times
\mathcal I_{\mathrm{comm}}^t
\right),
\quad
\mathcal N_i^{\mathrm{out},t}
=
\{j:(i,j)\in\mathcal E_{\mathrm{comm}}^t\},
\quad
\mathcal N_i^{\mathrm{in},t}
=
\{j:(j,i)\in\mathcal E_{\mathrm{comm}}^t\}
\label{eq:app_induced_neighborhoods}
\end{equation}
denote the directed, self-loop-free edge set induced on
$\mathcal I_{\mathrm{comm}}^t$ and the associated neighborhoods. Agents that
exchange no messages at round $t$ influence one another only through the
shared team reward, so modelling the interaction terms of the reference value
as supported on $\mathcal E_{\mathrm{comm}}^t$ is the natural structural
assumption.

\begin{proposition}[Graph-supported coefficients]
\label{prop:acts_graph_scopes}
Let $\mathcal S(G^t)$ be a collection of scopes
$S\subseteq\mathcal I_{\mathrm{comm}}^t$ supported on $G^t$, write
$\boldsymbol q_S=(q_j)_{j\in S}$, and suppose
\begin{equation}
F_{\mathcal C^t}(\boldsymbol q)
=
a_{\mathcal C^t}
+
\sum_{S\in\mathcal S(G^t)}
f_S(\boldsymbol q_S;G^t),
\label{eq:app_general_graph_factorization}
\end{equation}
with every $f_S$ continuously differentiable in its arguments. Then
\begin{equation}
\beta_i^{\mathrm{ref},t}(\boldsymbol q)
=
\sum_{\substack{S\in\mathcal S(G^t)\\ i\in S}}
\partial_if_S(\boldsymbol q_S;G^t).
\label{eq:app_general_graph_coefficient}
\end{equation}
\end{proposition}

\begin{proof}
Differentiating Equation~\ref{eq:app_general_graph_factorization} with respect
to $q_i$ annihilates $a_{\mathcal C^t}$ and every factor whose scope does not
contain $i$, since such factors do not depend on $q_i$. Summing the
derivatives of the remaining factors gives
Equation~\ref{eq:app_general_graph_coefficient}.
\end{proof}

Proposition~\ref{prop:acts_graph_scopes} places no restriction on the
cardinality of the scopes, so every graph-induced interaction factor
containing agent $i$, including higher-order ones, can contribute to its
coefficient. The pairwise case isolates the two directed contributions.

\begin{corollary}[Directed-neighborhood dependence]
\label{cor:acts_neighborhood_coefficients}
Suppose
\begin{equation}
F_{\mathcal C^t}(\boldsymbol q)
=
a_{\mathcal C^t}
+
\sum_{i\in\mathcal I_{\mathrm{comm}}^t}
f_i(q_i)
+
\sum_{(i,j)\in\mathcal E_{\mathrm{comm}}^t}
f_{ij}(q_i,q_j).
\label{eq:app_graph_factorization}
\end{equation}
Then
\begin{equation}
\begin{aligned}
\beta_i^{\mathrm{ref},t}(\boldsymbol q)
={}&
f_i'(q_i)
+
\sum_{j\in\mathcal N_i^{\mathrm{out},t}}
\partial_1f_{ij}(q_i,q_j)
\\
&+
\sum_{j\in\mathcal N_i^{\mathrm{in},t}}
\partial_2f_{ji}(q_j,q_i).
\end{aligned}
\label{eq:app_neighbor_coefficient}
\end{equation}
\end{corollary}

\begin{proof}
Apply Proposition~\ref{prop:acts_graph_scopes} with
$\mathcal S(G^t)
=\{\{i\}:i\in\mathcal I_{\mathrm{comm}}^t\}
\cup\mathcal E_{\mathrm{comm}}^t$.
The unary scope containing $i$ contributes $f_i'(q_i)$. Because
$\mathcal E_{\mathrm{comm}}^t$ is directed and self-loop free, a pairwise
scope contains $i$ in exactly one of two ways: for an outgoing edge $(i,j)$
the utility $q_i$ occupies the first argument of $f_{ij}$, and for an
incoming edge $(j,i)$ it occupies the second argument of $f_{ji}$. Summing
the corresponding partial derivatives gives
Equation~\ref{eq:app_neighbor_coefficient}.
\end{proof}

Corollary~\ref{cor:acts_neighborhood_coefficients} is the design
justification for the topology encoder of
Equation~\ref{eq:acts_topology_encoder}. The incoming and outgoing
neighborhoods enter $\beta_i^{\mathrm{ref},t}$ through distinct partial
derivatives, $\partial_2f_{ji}$ and $\partial_1f_{ij}$ respectively, so a
representation that merges them into a single undirected neighborhood
conflates two terms that are not interchangeable. A representation built from
graph-level statistics such as edge count or density is weaker still: it
discards the identity of agent $i$ and therefore cannot express the
restriction to scopes containing $i$ that
Equation~\ref{eq:app_general_graph_coefficient} imposes. \texttt{ACTS}
accordingly encodes both directed neighborhoods in $z_i^t$. It does not
evaluate $\beta_i^{\mathrm{ref},t}$ analytically; it learns a nonnegative
surrogate $\widehat\beta_i^t$ from TD targets, as described next.

\subsection{Topology-Aware TD Credit Assignment}
\label{app:topology_aware_td_credit}

For $i\in\mathcal I_{\mathrm{comm}}^t$ define the individual utility and the
topological-attention logit of head $h$,
\begin{equation}
Q_i^t(a_i)
=
Q_i(\xi_i^t,a_i),
\qquad
\eta_{i,h}^t
=
\frac{
\left(W_{Q,\mathrm{top}}^{(h)}s^t\right)^\top
\left(W_{K,\mathrm{top}}^{(h)}z_i^t\right)
}{\sqrt{d_h}},
\label{eq:app_utility_and_logit}
\end{equation}
with weights normalized over the eligible agents,
\begin{equation}
\lambda_{i,h}^t
=
\frac{\exp\!\left(\eta_{i,h}^t\right)}
{\sum_{j\in\mathcal I_{\mathrm{comm}}^t}
\exp\!\left(\eta_{j,h}^t\right)} .
\label{eq:app_attention_weight}
\end{equation}
Let $w_h^t=\left|\left[f_{\theta_w}(s^t)\right]_h\right|\geq0$ be the
nonnegative head weights and $b^t=V_{\theta_b}(s^t)$ the state-dependent
bias. The head-level and joint action values are
\begin{equation}
\begin{aligned}
Q_h^t(\boldsymbol a)
&=
\sum_{i\in\mathcal I_{\mathrm{comm}}^t}
\lambda_{i,h}^tQ_i^t(a_i),
\\
Q_{\mathrm{tot}}^t(\mathcal C^t,\boldsymbol a)
&=
b^t
+
\sum_{h=1}^{H}w_h^tQ_h^t(\boldsymbol a)
=
b^t
+
\sum_{i\in\mathcal I_{\mathrm{comm}}^t}
\widehat\beta_i^tQ_i^t(a_i),
\end{aligned}
\label{eq:app_acts_mixer}
\end{equation}
where the effective mixing coefficient is
\begin{equation}
\widehat\beta_i^t
=
\sum_{h=1}^{H}w_h^t\lambda_{i,h}^t
\;\geq\;0 .
\label{eq:app_learned_coefficient}
\end{equation}
The second equality in Equation~\ref{eq:app_acts_mixer} follows by exchanging
the two finite sums, and shows that the multi-head mixer is, for a fixed
context, a nonnegatively weighted sum of individual utilities whose weights
are determined by $\{z_i^t\}$ and hence by $G^t$.

\begin{remark}[Functional form of a head]
\label{rem:acts_head_direction}
Since $s^t$ is shared by all agents within a head,
Equation~\ref{eq:app_utility_and_logit} can be written as
$\eta_{i,h}^t=\langle v_h(s^t),z_i^t\rangle$ with
$v_h(s^t)
=(W_{K,\mathrm{top}}^{(h)})^\top
W_{Q,\mathrm{top}}^{(h)}s^t/\sqrt{d_h}$.
Each head therefore contributes one state-dependent linear direction in
neighborhood-embedding space along which agents are compared, and $w_h^t$
determines how much that direction matters in the current state.
\end{remark}

For a sampled transition let $y^t$ denote the detached TD target and set
\begin{equation}
\delta^t
=
Q_{\mathrm{tot}}^t(\mathcal C^t,\boldsymbol a^t)-y^t,
\qquad
\ell_{\mathrm{TD}}^t
=
\tfrac12\left(\delta^t\right)^2 .
\label{eq:acts_app_td_loss}
\end{equation}

\begin{proposition}[Topology-aware TD gradients]
\label{prop:acts_td_gradients}
Conditional on the context $\mathcal C^t$ and with $y^t$ treated as a
constant,
\begin{equation}
\frac{\partial\ell_{\mathrm{TD}}^t}
{\partial Q_i^t(a_i^t)}
=
\delta^t\widehat\beta_i^t,
\qquad
\frac{\partial\ell_{\mathrm{TD}}^t}
{\partial w_h^t}
=
\delta^tQ_h^t(\boldsymbol a^t),
\label{eq:app_utility_gradient}
\end{equation}
and for every eligible agent $i$ and head $h$,
\begin{equation}
\frac{\partial\ell_{\mathrm{TD}}^t}
{\partial\eta_{i,h}^t}
=
\delta^tw_h^t\lambda_{i,h}^t
\left(
Q_i^t(a_i^t)-Q_h^t(\boldsymbol a^t)
\right).
\label{eq:app_attention_gradient}
\end{equation}
\end{proposition}

\begin{proof}
~\eqref{eq:app_acts_mixer} gives
$\partial Q_{\mathrm{tot}}^t/\partial Q_i^t(a_i^t)=\widehat\beta_i^t$ and
$\partial Q_{\mathrm{tot}}^t/\partial w_h^t=Q_h^t(\boldsymbol a^t)$; since
$\partial\ell_{\mathrm{TD}}^t/\partial Q_{\mathrm{tot}}^t=\delta^t$ when
$y^t$ is detached, the chain rule yields
Equation~\ref{eq:app_utility_gradient}. For the logits, the softmax Jacobian of
Equation~\ref{eq:app_attention_weight} is
\[
\frac{\partial\lambda_{j,h}^t}{\partial\eta_{i,h}^t}
=
\lambda_{j,h}^t
\left(
\mathbb I[i=j]-\lambda_{i,h}^t
\right),
\]
and $\eta_{i,h}^t$ enters only head $h$, so
\[
\frac{\partial Q_h^t(\boldsymbol a^t)}{\partial\eta_{i,h}^t}
=
\sum_{j\in\mathcal I_{\mathrm{comm}}^t}
Q_j^t(a_j^t)\lambda_{j,h}^t
\left(
\mathbb I[i=j]-\lambda_{i,h}^t
\right)
=
\lambda_{i,h}^t
\left(
Q_i^t(a_i^t)-Q_h^t(\boldsymbol a^t)
\right),
\]
where the last step uses
$\sum_j\lambda_{j,h}^tQ_j^t(a_j^t)=Q_h^t(\boldsymbol a^t)$. Combining with
$\partial Q_{\mathrm{tot}}^t/\partial Q_h^t=w_h^t$ proves
Equation~\ref{eq:app_attention_gradient}.
\end{proof}

~\eqref{eq:app_utility_gradient} shows that $\widehat\beta_i^t$ is
simultaneously the Jacobian of the mixer with respect to agent $i$'s selected
utility and the scale of the direct TD signal that this utility receives:
monotonicity and credit assignment are governed by the same quantity. Because
$\widehat\beta_i^t$ depends on $z_i^t$, and $z_i^t$ encodes the directed
neighborhoods of $G^t$ by
Corollary~\ref{cor:acts_neighborhood_coefficients}, the allocation is
topology aware. ~\eqref{eq:app_attention_gradient} has the form of an
advantage: the deviation $Q_i^t-Q_h^t$ compares agent $i$'s utility against
the weighted mean of head $h$, so a head increases its weight on agents whose
utility exceeds that mean whenever the joint value must increase
($\delta^t<0$), and decreases it otherwise. The head weights are produced by
an absolute value and are therefore differentiable except on the measure-zero
set $[f_{\theta_w}(s^t)]_h=0$, where the subgradient is taken to be zero, as
is standard for monotonic mixers \citep{qmix}.

\begin{remark}[Where credit for partner selection enters]
\label{rem:acts_credit_path}
All quantities in Equation~\ref{eq:app_acts_mixer} are computed from the
pre-decision context, so $Q_{\mathrm{tot}}^t$ does not depend on the topology
$G^{t+1}$ that $\boldsymbol a^t$ induces, and the discrete map $\Psi$ of
Equation~\ref{eq:app_topology_transition} is never differentiated. Credit for a
partner configuration is consequently not delivered within round $t$; it is
delivered through the TD target, because $Q_{\mathrm{tot}}^{t+1}$ is
conditioned on $\mathcal C^{t+1}$ and hence on $G^{t+1}$. A partner set that
improves the team's subsequent value raises $y^t$ and, through
Equation~\ref{eq:app_utility_gradient}, increases $Q_i^t(a_i^t)$ for the actions
that produced it. Relational attention additionally receives gradient within
round $t$ along a differentiable path, since the aggregated feature $\xi_i^t$
of Equation~\ref{eq:acts_cross_attention_selection} enters $Q_i^t$ directly.
\end{remark}

\subsection{Monotonicity and Individual--Global--Max Consistency}
\label{app:acts_igm}

Let each eligible agent have a finite, nonempty action space
$\mathcal A_i^t\subseteq\mathcal A_{\mathrm{comm}}$ and write
\begin{equation}
\mathcal A^t
=
\prod_{i\in\mathcal I_{\mathrm{comm}}^t}
\mathcal A_i^t .
\label{eq:app_joint_action_space}
\end{equation}
The context and the mixing coefficients are held fixed while the elements of
$\mathcal A^t$ are compared. In particular $a^{\mathrm{off}}$ and
$a^{\mathrm{solo}}$ are ordinary candidates: their effect is realized in
$\mathcal I_{\mathrm{comm}}^{t+1}$ and $G^{t+1}$ through
Equation~\ref{eq:app_topology_transition}, not in the current attention
normalization. By Equation~\ref{eq:app_learned_coefficient},
\begin{equation}
\frac{\partial Q_{\mathrm{tot}}^t}
{\partial Q_i^t(a_i)}
=
\widehat\beta_i^t
\;\geq\;0
\qquad
\forall i\in\mathcal I_{\mathrm{comm}}^t,
\label{eq:acts_app_monotonicity}
\end{equation}
so the mixer is monotonic in every individual utility. Define
\begin{equation}
\mathcal M_i^t
=
\operatorname*{arg\,max}_{a_i\in\mathcal A_i^t}
Q_i^t(a_i),
\qquad
\mathcal P^t
=
\left\{i:\widehat\beta_i^t>0\right\},
\qquad
\mathcal Z^t
=
\left\{i:\widehat\beta_i^t=0\right\}.
\label{eq:app_argmax_sets}
\end{equation}

\begin{theorem}[IGM consistency]
\label{thm:acts_igm}
For the fixed context $\mathcal C^t$,
\begin{equation}
\operatorname*{arg\,max}_{\boldsymbol a\in\mathcal A^t}
Q_{\mathrm{tot}}^t(\mathcal C^t,\boldsymbol a)
=
\left(
\prod_{i\in\mathcal P^t}\mathcal M_i^t
\right)
\times
\left(
\prod_{i\in\mathcal Z^t}\mathcal A_i^t
\right),
\label{eq:app_general_argmax}
\end{equation}
so every independently greedy joint action is a joint maximizer. If in
addition the attention logits are finite and
\begin{equation}
W^t
:=
\sum_{h=1}^{H}w_h^t
\;>\;0,
\label{eq:app_positive_head_mass}
\end{equation}
then $\widehat\beta_i^t>0$ for every $i\in\mathcal I_{\mathrm{comm}}^t$ and
the set-valued IGM identity holds exactly:
\begin{equation}
\operatorname*{arg\,max}_{\boldsymbol a\in\mathcal A^t}
Q_{\mathrm{tot}}^t(\mathcal C^t,\boldsymbol a)
=
\prod_{i\in\mathcal I_{\mathrm{comm}}^t}
\operatorname*{arg\,max}_{a_i\in\mathcal A_i^t}
Q_i(\xi_i^t,a_i).
\label{eq:app_exact_igm}
\end{equation}
\end{theorem}

\begin{proof}
Let $m_i^t=\max_{a_i\in\mathcal A_i^t}Q_i^t(a_i)$, which exists because
$\mathcal A_i^t$ is finite and nonempty. Since $\widehat\beta_i^t\geq0$ by
Equation~\ref{eq:app_learned_coefficient}, each summand of
Equation~\ref{eq:app_acts_mixer} obeys
$\widehat\beta_i^tQ_i^t(a_i)\leq\widehat\beta_i^tm_i^t$, hence
\begin{equation}
Q_{\mathrm{tot}}^t(\mathcal C^t,\boldsymbol a)
\;\leq\;
b^t
+
\sum_{i\in\mathcal I_{\mathrm{comm}}^t}
\widehat\beta_i^tm_i^t
\label{eq:app_igm_upper_bound}
\end{equation}
for every $\boldsymbol a\in\mathcal A^t$, and the bound is attained. Equality
in Equation~\ref{eq:app_igm_upper_bound} holds if and only if
$Q_i^t(a_i)=m_i^t$, that is $a_i\in\mathcal M_i^t$, for every
$i\in\mathcal P^t$, while for $i\in\mathcal Z^t$ the term
$\widehat\beta_i^tQ_i^t(a_i)$ vanishes and the choice of $a_i$ leaves
$Q_{\mathrm{tot}}^t$ unchanged. This establishes
Equation~\ref{eq:app_general_argmax}, and
$\mathcal M_i^t\subseteq\mathcal A_i^t$ gives the subsequent claim. For the
second part, finite logits make the softmax of
Equation~\ref{eq:app_attention_weight} strictly positive, so
$\lambda_{i,h}^t>0$ for all $i$ and $h$. If $W^t>0$ there exists a head
$h^\star$ with $w_{h^\star}^t>0$, whence
$\widehat\beta_i^t
\geq w_{h^\star}^t\lambda_{i,h^\star}^t>0$
for every eligible agent. Therefore $\mathcal Z^t=\varnothing$ and
Equation~\ref{eq:app_general_argmax} reduces to
Equation~\ref{eq:app_exact_igm}.
\end{proof}

If all head weights vanish, then $Q_{\mathrm{tot}}^t=b^t$ for every candidate
joint action; independently greedy actions remain joint maximizers by
Equation~\ref{eq:app_general_argmax}, but the two maximizing sets need not
coincide. The condition $W^t>0$ of
Equation~\ref{eq:app_positive_head_mass} can be enforced through a strictly
positive head-weight parameterization whenever the set-valued identity is
required.

\begin{remark}[Reachable coefficients and the cost of monotonicity]
\label{rem:acts_reachable_coefficients}
Because $\lambda_{\cdot,h}^t$ is normalized over
$\mathcal I_{\mathrm{comm}}^t$, the learned coefficients satisfy
\begin{equation}
\sum_{i\in\mathcal I_{\mathrm{comm}}^t}
\widehat\beta_i^t
=
\sum_{h=1}^{H}w_h^t
\sum_{i\in\mathcal I_{\mathrm{comm}}^t}
\lambda_{i,h}^t
=
W^t .
\label{eq:app_coefficient_mass}
\end{equation}
The mixer therefore redistributes a single state-dependent mass $W^t$ across
agents rather than rescaling each utility independently, and by
Equation~\ref{eq:acts_app_monotonicity} it cannot represent the negative
sensitivities that Equation~\ref{eq:app_neighbor_coefficient} permits. Both
restrictions are the price of monotonicity, which is what makes
Equation~\ref{eq:app_exact_igm} and hence decentralized greedy execution
available \citep{qmix}. What $\widehat\beta_i^t$ adds relative to a
state-only mixer is not expressiveness in the utilities but the dependence of
the normalized credit profile on $\mathcal N_i^{\mathrm{in},t}$ and
$\mathcal N_i^{\mathrm{out},t}$ established in
Corollary~\ref{cor:acts_neighborhood_coefficients}; the empirical claim of
this paper is that this monotonic, topology-conditioned class suffices for the
communication action space of
Section~\ref{sec:topology_learning}, not that it is universal.
\end{remark}

\begin{remark}[Scope of the consistency result]
\label{rem:acts_igm_scope}
Theorem~\ref{thm:acts_igm} relies on the temporal ordering of
Eqs.~\eqref{eq:app_context}--\eqref{eq:app_topology_transition}: all
candidate actions at round $t$ share one pre-decision context and therefore
one set of mixing coefficients, and the selected action determines only the
next topology. Without this ordering---for instance if the mixer were
conditioned on the topology induced by the candidate action itself---the
coefficients would vary across $\mathcal A^t$, the decomposition in
Equation~\ref{eq:app_acts_mixer} would no longer be a fixed weighted sum, and
monotonicity in each $Q_i^t$ would not imply
Equation~\ref{eq:app_exact_igm}.
\end{remark}

\section{ACTS Training and Decentralized Execution}
\label{app:acts_pseudocode}

\begin{algorithm}[H]
\caption{\texttt{ACTS} Centralized Training}
\label{alg:acts_training}
\footnotesize
\begin{algorithmic}[1]
\setlength{\itemsep}{0pt}

\Require Training tasks $\mathcal D_{\mathrm{train}}$; roles
$\{\rho_i\}_{i=1}^{N}$; episodes $M$; horizon $T$; discount factor
$\gamma$; objective weights $w_{\mathrm{task}},w_{\mathrm{tok}}$;
token scale $B$; initial exploration rate $\epsilon_0$; minibatch size
$b$; target update interval $K$

\Ensure Online parameters $\Theta$ and target parameters $\Theta^{-}$

\State Initialize $\Theta$, $\Theta^{-}\leftarrow\Theta$,
replay buffer $\mathcal B\leftarrow\varnothing$,
$\epsilon\leftarrow\epsilon_0$, and update counter $u\leftarrow0$

\For{$e=1,\ldots,M$}

    \State Sample $x\sim\mathcal D_{\mathrm{train}}$;
    initialize $G^1$, $\mathcal I_{\mathrm{comm}}^1=\mathcal I$,
    $h_i^0=\mathbf 0$, trajectory $\tau\leftarrow(\,)$,
    and $C_{\mathrm{tok}}\leftarrow0$

    \For{$t=1,\ldots,T$}

        \State Encode the pre-decision context
        $\mathcal C^t$ using Equations
        \ref{eq:acts_state_encoding}--\ref{eq:acts_topology_encoder}

        \State Select $\boldsymbol a^t$ with $\epsilon$-greedy local
        utilities; assign $a_i^t=a^{\mathrm{off}}$ to inactive agents

        \State Construct $\mathcal I_{\mathrm{comm}}^{t+1}$ and $G^{t+1}$
        using Equations
        \ref{eq:acts_communication_eligible_set} and
        \ref{eq:acts_induced_topology}

        \State Set
        $d^t\leftarrow
        \mathbf 1[
        |\mathcal I_{\mathrm{comm}}^{t+1}|<\lceil N/2\rceil
        \ \lor\ t=T]$

        \State Execute round $t$ and update $C_{\mathrm{tok}}$;
        set $r^t\leftarrow0$

        \If{$d^t=0$}
            \State Observe $s^{t+1}$ and $\{o_i^{t+1}\}$
        \Else
            \State Set terminal next-state fields to $\bot$
        \EndIf

        \State Append
        $\zeta^t=(s^t,\boldsymbol o^t,G^t,
        \mathcal I_{\mathrm{comm}}^t,\boldsymbol a^t,r^t,
        s^{t+1},\boldsymbol o^{t+1},G^{t+1},
        \mathcal I_{\mathrm{comm}}^{t+1},d^t)$
        to $\tau$

        \If{$d^t=1$}
            \State \textbf{break}
        \EndIf

    \EndFor

    \State Aggregate the final answer and compute
    $R\leftarrow
    w_{\mathrm{task}}R_{\mathrm{task}}
    -w_{\mathrm{tok}}\min(C_{\mathrm{tok}}/B,1)$

    \State Set the reward of the final transition in $\tau$ to $R$
    and add all transitions in $\tau$ to $\mathcal B$

    \If{$|\mathcal B|\geq b$}

        \State Sample a minibatch
        $\{\zeta_k\}_{k=1}^{b}$ from $\mathcal B$

        \State Compute
        \[
        y_k=
        \begin{cases}
        r_k, & d_k=1,\\
        r_k+\gamma
        Q_{\mathrm{tot}}^{-}
        (\boldsymbol a_k^{*,+};\mathcal C_k^{+}),
        & d_k=0,
        \end{cases}
        \]
        where $\boldsymbol a_k^{*,+}$ is selected using the target
        local utilities

        \State Update $\Theta$ by minimizing
        \[
        \mathcal L_{\mathrm{TD}}
        =
        \frac{1}{2b}
        \sum_{k=1}^{b}
        \left[
        y_k-
        Q_{\mathrm{tot}}(\boldsymbol a_k;\mathcal C_k)
        \right]^2
        \]
        with $y_k$ detached and the backbone LLM frozen

        \State Set $u\leftarrow u+1$

        \If{$u\bmod K=0$}
            \State $\Theta^{-}\leftarrow\Theta$
        \EndIf

    \EndIf

    \State Anneal $\epsilon$

\EndFor

\end{algorithmic}
\end{algorithm}

We reuse the GNN--GRU encoder and relational-attention/$Q$ branch of
Section~\ref{sec:topology_learning} and the topology encoder and topological
mixer of Section~\ref{sec:topological_factorization}. For
$G^t=(\mathcal I,\mathcal E^t)$, write
$\mathcal N_i^{\mathrm{in},t}=\{j\mid(j,i)\in\mathcal E^t\}$ and
$\mathcal N_i^{\mathrm{out},t}=\{j\mid(i,j)\in\mathcal E^t\}$ for the directed
neighborhoods encoded by $z_i^t$ in
~\eqref{eq:acts_topology_encoder}. The operator $\textsc{BuildTopology}$ denotes the topology construction map
defined in Appendix~\ref{app:directed_topology_construction}. The six
actions are listed in Appendix~\ref{sec:appendix_action_space}, and the
observation and state vectors are specified in
Appendix~\ref{sec:appendix_states}. The action $a^{\mathrm{off}}$ suppresses
both backbone inference and communication. The action $a^{\mathrm{solo}}$
invokes one independent backbone inference without creating communication
edges, and its output is retained for final answer aggregation. The
corresponding agent is excluded from subsequent communication rounds. All
other actions produce an output and induce communication edges according to
the selected mode. These messages enter the observations of round $t+1$.

Every round is organized around the pre-decision context
$\mathcal C^t=(s^t,G^t,\mathcal I_{\mathrm{comm}}^t,
\{\xi_i^t,z_i^t\}_{i\in\mathcal I_{\mathrm{comm}}^t})$, computed from the
realized topology before any action is selected and held fixed while candidate
actions are compared, so the coefficients
$\widehat\beta_i^t=\sum_{h}w_h^t\lambda_{i,h}^t$ are constants of
$\mathcal C^t$ and ~\eqref{eq:acts_igm} applies. Since agents outside $\mathcal I_{\mathrm{comm}}^t$ are assigned
$a^{\mathrm{off}}$, ~\eqref{eq:acts_communication_eligible_set}
gives
$\mathcal I_{\mathrm{comm}}^{t+1}\subseteq\mathcal I_{\mathrm{comm}}^t$.
An agent selecting $a^{\mathrm{off}}$ or $a^{\mathrm{solo}}$ therefore
withdraws from subsequent communication, although only
$a^{\mathrm{off}}$ suppresses the current-round backbone inference. Hard partner-set construction is nondifferentiable, so
relational attention is supervised only through the soft feature $\xi_i^t$
while topological attention is trained through the mixer; the target branch is
detached, and replay stores raw inputs and topologies rather than cached
hidden states.

\begin{algorithm}[H]
\caption{\texttt{ACTS} Decentralized Execution}
\label{alg:acts_execution}
\footnotesize
\begin{algorithmic}[1]
\setlength{\itemsep}{0pt}
\Require Query $x$; roles $\{\rho_i\}_{i=1}^{N}$; horizon $T$; trained
parameters $\Theta$
\Ensure Team answer $Y$
\State $G^1\!\leftarrow$ complete digraph on $\mathcal I$ without self-loops;
$\mathcal I_{\mathrm{comm}}^1\!\leftarrow\!\mathcal I$;
$h_i^0\!\leftarrow\!\mathbf 0$
\For{$t=1,\ldots,T$}
  \State Build $o_i^t$ from $x$, agent $i$'s history, and the messages
  delivered under $G^t$
  \State For $i\!\in\!\mathcal I_{\mathrm{comm}}^t$: $h_i^t$, $\xi_i^t$,
  $\{\alpha_{ij}^t\}$, then $a_i^t\!\leftarrow\!
  \operatorname*{arg\,max}_{a\in\mathcal A_{\mathrm{comm}}}Q_i(\xi_i^t,a)$
  \Comment{Equation~\ref{eq:acts_state_encoding}--\eqref{eq:acts_cross_attention_selection},~\eqref{eq:acts_igm}}
  \State $a_i^t\!\leftarrow\!a^{\mathrm{off}}$ for
  $i\!\notin\!\mathcal I_{\mathrm{comm}}^t$; form $\boldsymbol a^t$ and
  $\boldsymbol\alpha^t$
  \State $\mathcal I_{\mathrm{comm}}^{t+1}\!\leftarrow\!
  \{i\mid a_i^t\!\notin\!\{a^{\mathrm{off}},a^{\mathrm{solo}}\}\}$;
  $G^{t+1}\!\leftarrow\!\textsc{BuildTopology}(\boldsymbol a^t,
  \boldsymbol\alpha^t)$
  \Comment{Equation~\ref{eq:acts_communication_eligible_set},~\eqref{eq:acts_induced_topology}}
  \State Run round $t$ ($a^{\mathrm{off}}$: no call; $a^{\mathrm{solo}}$: own
  answer; else message along $G^{t+1}$) and append the outputs to the local
  histories
  \State \textbf{if} $|\mathcal I_{\mathrm{comm}}^{t+1}|\!<\!\lceil N/2\rceil$
  \textbf{then break}
\EndFor
\State \Return $Y$ from the answer-aggregation node
\end{algorithmic}
\end{algorithm}

\section{Experimental Details}
\label{sec:appendix_experiment}

\subsection{Training and Evaluation Details}
\label{sec:appendix_training}

% All methods share the same team size ($N=3$), task prompts, and communication-round
% budget (three rounds for mathematical tasks, two rounds for coding and reasoning
% tasks). The scheduler is trained with the Adam optimizer at learning rate
% $5\times10^{-4}$, discount factor $\gamma=0.99$, replay buffer capacity
% $5{,}000$, batch size $8$, and $\epsilon$-greedy exploration that decays
% linearly from $1.0$ to $0.05$ over $50$ episodes. Each domain provides
% $15$ training examples. Results are averaged over three independent runs
% with different random seeds.
\paragraph{Benchmarks and backbone models.}
We evaluate \texttt{ACTS} on nine benchmarks across four domains. The coding
benchmarks are LiveCodeBench v6 and HumanEval
\citep{jain2024livecodebench,chen2021humaneval}. The reasoning benchmarks
are MMLU-Pro and the first 250 questions of Humanity's Last Exam
\citep{wang2024mmlupro,hle2025}. HotpotQA is used for multi-hop question
answering \citep{yang2018hotpotqa}. The mathematical benchmarks are AIME 2025, AIME 2026, HMMT 2025,
and Beyond-AIME
\citep{maa2025aime,maa2026aime,matharena2025hmmt,bytedance2025beyondaime}.
Each AIME benchmark comprises 30 problems from the corresponding
AIME I and AIME II examinations.
Experiments use GPT-OSS-120B and DeepSeek-V4-Flash as frozen
backbone LLMs.
Experiments use GPT-OSS-120B and DeepSeek-V4-Flash as frozen backbone LLMs.
Each task is instantiated with three role-conditioned agents followed by a
final answer-aggregation node.

\paragraph{Baselines and controlled comparison.}
We compare \texttt{ACTS} with direct prompting, LLM-Debate
\citep{llm_debate}, AutoGen \citep{autogen}, LangGraph
\citep{langchain2024langgraph},  Microsoft Agent Framework \citep{microsoft2025agentframework}, TopoDIM \citep{topodim}, and Agent Q-Mix
\citep{agent_qmix}. All methods use the same team size
($N=3$), task prompts, communication-round budget, and answer-evaluation
procedure.

\paragraph{Training protocol and evaluation.}
The ACTS scheduler is trained for 50 episodes using 15 training examples
per domain. The communication budget is three rounds for mathematical tasks
and two rounds for coding and reasoning tasks. All trainable modules are
optimized with Adam using learning rate $5\times10^{-4}$ and discount factor
$\gamma=0.99$. The replay buffer has capacity $5{,}000$, the minibatch size
is $8$, and $\epsilon$-greedy exploration is annealed linearly from $1.0$ to
$0.05$ over 50 episodes. The learned parameters are fixed during evaluation.
Reported results are averaged over three independent runs with different
random seeds.

\subsubsection{Data Sampling and Information Gain Computation}
\label{sec:ig_and_dataset}

For each model $M$, benchmark $D$, and topology family $C$, empirical
communication gain is defined as
$\Delta(M,D,C)=
\operatorname{mean}_{m\in C}\operatorname{ACC}(m,M,D)
-\operatorname{ACC}(\mathrm{Direct},M,D)$,
with accuracy expressed as a fraction.
The threshold $x=2/n$ is a heuristic tolerance corresponding to two
correct answers among $n$ evaluation samples.
A positive gain is identified when at least one topology family
satisfies $\Delta(M,D,C)>x$.
Otherwise, no positive gain is identified under this criterion.
This descriptive threshold does not establish statistical significance.

\begin{table}[t]
\centering
\caption{Per-benchmark evaluation samples, IG thresholds, and classifications under GPT-OSS-120B and DeepSeek-V4-Flash.}
\label{tab:ig_dataset}

\footnotesize
\setlength{\tabcolsep}{5pt}
\renewcommand{\arraystretch}{1.12}

\begin{tabular}{l l c c c c}
\toprule
Benchmark & Domain & $n$ & $x = 2/n$ & IG (GPT-OSS-120B) & IG (DeepSeek-V4-Flash) \\
\midrule
LiveCodeBench & Code & 175 & 0.0114 & $>0$ & $>0$ \\
HumanEval & Code & 164 & 0.0122 & $\approx 0$ & $\approx 0$ \\
HLE & Reasoning & 250 & 0.0080 & $\approx 0$ & $\approx 0$ \\
HotpotQA & Multi-hop QA & 740 & 0.0027 & $>0$ & $>0$ \\
MMLU-Pro & Reasoning & 70 & 0.0286 & $>0$ & $>0$ \\
AIME 2025 & Math & 30 & 0.0667 & $>0$ & $>0$ \\
AIME 2026 & Math & 30 & 0.0667 & $>0$ & $>0$ \\
Beyond-AIME & Math & 100 & 0.0200 & $>0$ & $>0$ \\
HMMT 2025 & Math & 30 & 0.0667 & $>0$ & $\approx 0$ \\
\bottomrule
\end{tabular}

\end{table}

\subsection{States and Observations}
\label{sec:appendix_states}

\texttt{ACTS} follows a centralised-training, decentralised-execution scheme. A
shared per-agent value network receives a local observation $o_i$ for each
agent $i$, while a mixing network integrates the agent-specific values
into a joint estimate $Q_{\mathrm{tot}}$ conditioned on a global state
$s$. We describe the construction of each.

\paragraph{Local observation.}
For $N=3$, each agent receives a 51-dimensional observation $o_i^t$
comprising a 32-dimensional task embedding, a 16-dimensional one-hot
agent identifier, a binary indicator of whether the agent has produced
an answer, its incident edge count, and the normalized cumulative
team token consumption clipped to $[0,1]$.

\paragraph{Global state.}
The global state $s^t$ concatenates all agents' local observations,
three topology statistics comprising the edge count, graph density,
and agent count, and the flattened $N\times N$ relational-attention
matrix. For $N=3$, the resulting state dimension is
$3\times51+3+9=165$.

\subsection{Action Space}
\label{sec:appendix_action_space}

At each round every active agent selects one of six discrete actions,
grouped into three roles. Three actions, \emph{off}, \emph{solo}, and
\emph{broadcast}, operate outside the attention matrix: off deactivates
the agent entirely for the round, emitting no edges and issuing no LLM
call, and once a majority of agents are inactive the episode terminates
early; solo produces an independent answer and emits no edges; broadcast
sends the agent's current output to every other agent. Two actions build
sparse directed edges from the learned attention matrix $a_\mathrm{edge}$:
\emph{attn-broadcast} opens an edge $i\!\to\!j$ when the attention that
agent $j$ pays to $i$ exceeds the row mean, and \emph{attn-receive}
requests information only from the neighbours that agent $i$ itself
attends to. The remaining action, \emph{receive}, performs a generic
information fetch. Figure~\ref{fig:action_heatmap} reports the learned
distribution of these six actions across the nine benchmarks.

\begin{figure}[htbp]
\centering
\includegraphics[width=0.86\linewidth]{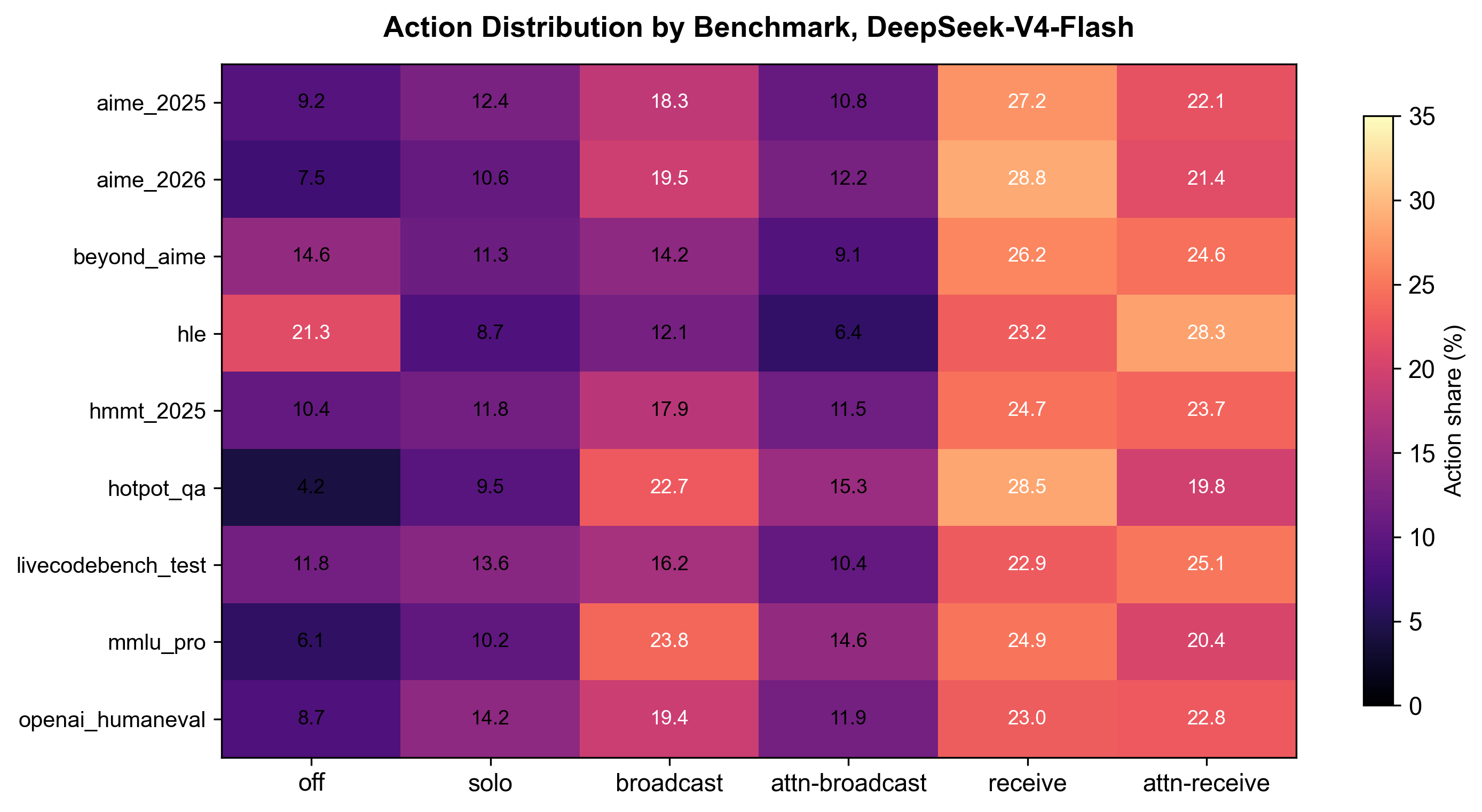}
\caption{Distribution of the six actions across benchmarks under
DeepSeek-V4-Flash. Each cell gives the proportion of rounds in which the
corresponding action is selected. The learned policy favours broadcast
and attn-broadcast on benchmarks where the agents contribute complementary
information, and shifts toward off and attn-receive where a single
perspective suffices.}
\label{fig:action_heatmap}
\end{figure}

\subsection{Case Study}
\label{sec:appendix_case_study}

Figure~\ref{fig:case_study} illustrates the same machinery at the
trajectory level on one HotpotQA question and one HLE question. On
HotpotQA the agents remain active throughout: several agents take the
broadcast action, sending their analysis of the question to the team
while one agent retains the solo action as an independent reference,
after which the broadcasting agents switch to attn-receive, so that the
initially connected topology contracts into a sparse pull while the
informative edges are preserved. On HLE one agent settles the answer
under the solo action, a single attn-broadcast check is routed only to
the agent that requires it, and the remaining agents are successively
deactivated through the off action, with the majority-inactive
condition terminating the episode early. The two trajectories
correspond to the two mechanisms that the aggregate action shares in
Figure~\ref{fig:action_parallel} average over: selective routing when
the agents contribute complementary information, and deactivation when
a single perspective suffices.

\begin{figure*}[htbp]
\centering
\includegraphics[width=0.92\linewidth]{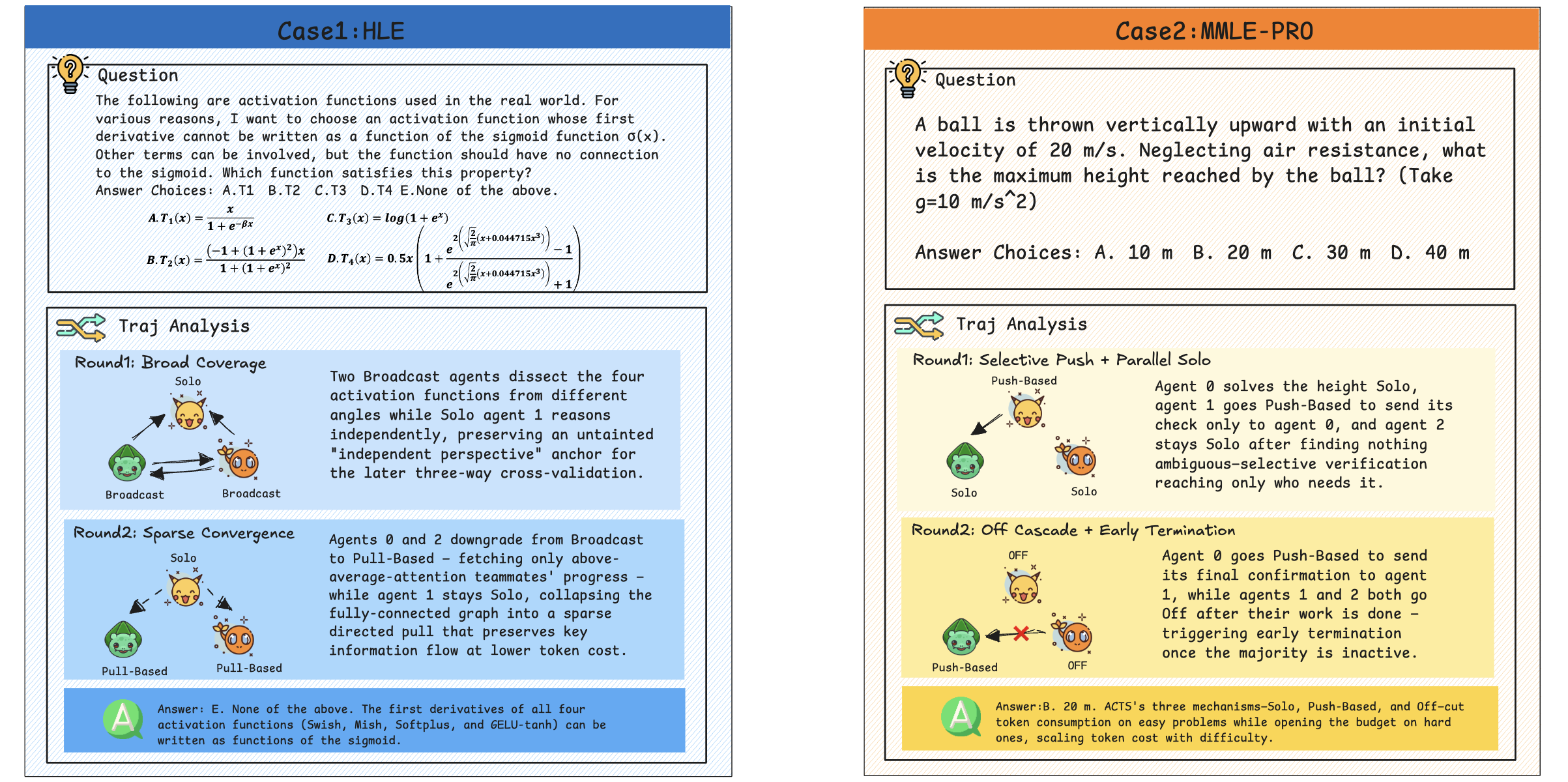}
\caption{Trajectory-level case study on one HotpotQA question and one HLE
question. On HotpotQA the connected topology contracts into a sparse pull
while one agent retains the solo action as an independent reference. On
HLE one agent settles the answer under the solo action, a single
attn-broadcast check is routed to the agent that requires it, and the
remaining agents are deactivated through the off action, triggering
early termination under the majority-inactive condition.}
\label{fig:case_study}
\end{figure*}

\subsection{Per-Sample Token Usage}
\label{sec:appendix_token}

Table~\ref{tab:token_usage} lists the average per-sample token consumption of
each method on HotpotQA, a high-information-gain ($\mathrm{IG}>0$) task, and
HLE, a low-information-gain ($\mathrm{IG}\approx 0$) task, for both backbone
models. These are the numbers plotted against accuracy in
Figure~\ref{fig:pareto}. Single-agent prompting uses the fewest tokens in
every column, but at the cost of substantially lower accuracy; \texttt{ACTS} consumes
the fewest tokens among all multi-agent methods in every column, while the
static topologies are the most expensive because they broadcast the full
running transcript to every agent each round.

\begin{table}[t]
\centering
\caption{Average per-sample token usage ($\times 10^{3}$) on HotpotQA, a
high-information-gain ($\mathrm{IG}>0$) task, and HLE, a
low-information-gain ($\mathrm{IG}\approx 0$) task, for both backbone
models. \texttt{ACTS} (highlighted) consumes the fewest tokens among all
multi-agent methods in every column; only single-agent prompting uses
fewer, at a large accuracy cost.}
\label{tab:token_usage}

\footnotesize
\setlength{\tabcolsep}{6pt}
\renewcommand{\arraystretch}{1.12}

\begin{tabular}{l cccc}
\toprule
& \multicolumn{2}{c}{HotpotQA} & \multicolumn{2}{c}{HLE} \\
\cmidrule(lr){2-3} \cmidrule(lr){4-5}
Method
& GPT-OSS-120B & DeepSeek-V4-Flash & GPT-OSS-120B & DeepSeek-V4-Flash \\
\midrule
Direct Prompting  & 5.8 & 3.6 & 4.5 & 2.8 \\
\midrule
LLM-Debate        & 14.1 & 9.5 & 12.3 & 8.2 \\
AutoGen           & 16.0 & 10.8 & 13.8 & 9.3 \\
LangGraph         & 16.4 & 11.0 & 14.0 & 9.5 \\
Agent Framework   & 16.7 & 11.3 & 14.2 & 9.7 \\
\midrule
G-Designer        & 12.5 & 8.4 & 10.8 & 7.2 \\
TopoDIM           & 11.4 & 7.6 & 9.8 & 6.6 \\
GTD               & 12.1 & 8.1 & 10.5 & 7.0 \\
MaAS              & 11.8 & 7.9 & 10.2 & 6.8 \\
\midrule
Agent Q-Mix       & 8.7 & 5.8 & 7.9 & 5.3 \\
\rowcolor{actsblue}
ACTS              & 7.1 & 4.8 & 6.4 & 4.3 \\
\bottomrule
\end{tabular}

\end{table}

\subsection{Agent Role Assignments and Prompts}
\label{sec:agent-prompts}

Each agent in the multi-agent system is assigned one of six predefined roles. Roles are allocated in round-robin fashion by agent index modulo~6. Since our experiments use $n=3$ agents, only the first three roles are instantiated in any given run:

\begin{enumerate}[label=\textbf{R\arabic*:}, leftmargin=*]
    \item \textbf{Independent Solver} --- solves the problem autonomously without relying on peer input.
    \item \textbf{Team Solver} --- approaches the problem collaboratively, integrating suggestions from teammates.
    \item \textbf{Questioner} --- probes assumptions and raises critical questions about the problem.
\end{enumerate}

The remaining three roles (Refiner, Verifier, Debater) are defined in the role pool for future extensibility to larger agent teams, or for experiments where $n>3$ (e.g., $n=6$). The round-robin allocation modulo~6 naturally activates them when the team size increases, without any change to the prompt construction logic.

\subsubsection{Agent System Prompt}

Every agent receives a system prompt that instantiates its role within an $n$-agent team and instructs it on the expected output format. The general template is:

\begin{quote}
\small
\texttt{You are \{ROLE\} on a \{$n$\}-agent team. \{SUFFIX\}}
\end{quote}

where \texttt{\{ROLE\}} is one of the six role labels listed above and \texttt{\{SUFFIX\}} is a question-type-specific instruction detailed below.

\subsubsection{Question-Type Suffixes}

The suffix appended to each agent's system prompt depends on the \texttt{answer\_type} of the question. Five answer types are supported:

\begin{description}
    \item[\texttt{numerical}.] \textit{``Solve the problem step by step, then state the final numerical answer. End with a line `Final Answer: \texttt{<number>}'.''}

    \item[\texttt{multiple\_choice}.] \textit{``This is a single-choice question; there is exactly one correct option. The options may be labeled with any letter from A up to Z (not just A--E), so read all options carefully before choosing. Briefly analyze why each candidate is right or wrong, then commit to exactly one option. End with a line `Final Answer: \texttt{<letter>}' (e.g., `Final Answer: G').''}

    \item[\texttt{text}.] \textit{``Answer concisely. If a number, give the number. If a name, give the name. End with a line `Final Answer: \texttt{<answer>}'.''}

    \item[\texttt{open\_answer}.] \textit{``Answer concisely in a few words. End with a line `Final Answer: \texttt{<answer>}'.''}

    \item[\texttt{code}.] \textit{``Write the complete, runnable solution code. Output ONLY the code---no markdown fences, no explanation.''}
\end{description}

\subsubsection{User Message}

The user message delivered to each agent is constructed by concatenating the following components in order:

\begin{enumerate}
    \item The question text.
    \item Answer options (only for \texttt{multiple\_choice} questions), formatted as \texttt{A. ...}, \texttt{B. ...}, etc.
    \item Reference passages, when provided by the dataset (e.g., supporting documents for multi-hop QA).
    \item Previously generated answers from inbound neighbors in the communication topology, listed as \texttt{Agent \{i\}: \{answer\}}.
\end{enumerate}

\subsubsection{Final Decision Maker}

After all agents have executed across all rounds, a \textbf{Final Decision Maker} reviews the collected answers and produces the single final output. Its system prompt is:

\begin{quote}
\small
\texttt{You are the Final Decision Maker of a \{$n$\}-agent team. Read the task and all agents' answers carefully, then give the correct single answer.}
\end{quote}

The corresponding user message includes the original question, options (if applicable), reference context, and the full list of agent answers, followed by the instruction \texttt{Give the final answer.}

\subsubsection{No-Role Baseline}

For ablation experiments, we also evaluate a variant where all agents share an identical, role-agnostic system prompt:

\begin{quote}
\small
\texttt{You are a general team member on a \{$n$\}-agent team. \{SUFFIX\}}
\end{quote}

This removes the role specialization signal while keeping all other aspects of the prompt (question-type suffix, user message structure, and communication topology) unchanged.

\end{document}